\documentclass{llncs}
\usepackage[T1]{fontenc}
\usepackage[backend=biber, style=alphabetic]{biblatex}
\usepackage{graphicx}
\usepackage{amsmath}
\usepackage{amssymb}
\usepackage{dsfont}
\usepackage{algorithm}
\usepackage[noend]{algpseudocode}
\usepackage{hyperref}
\usepackage{doi}
\usepackage{todonotes}
\usepackage[normalem]{ulem}
\usepackage{graphicx}
\usepackage{subcaption}
\usepackage{relsize}
\usepackage{minitoc}
\usepackage{tabularray}
\usepackage{bm}

\bibliography{main.bib}
\begin{document}
\title{AuthOr: Lower Cost Privacy-Free Garbling of Arbitrary Boolean Circuits}

%
\author{Osman Biçer \and Ali Ajorian}
\authorrunning{Biçer and Ajorian}
\institute{University of Basel, Switzerland}
\maketitle              % typeset the header of the contribution
\begin{abstract}
Privacy-free garbling schemes of Frederiksen et al.\ (Eurocrypt '15) are designed to satisfy only the authenticity criterion of Bellare et al.\ (ACM CCS '12) and to be more efficient than full-fledged garbling schemes.
In this work, we improve the state-of-the-art privacy-free version of half-gates (HG) garbling of Zahur et al.\ (Crypto '15) by allowing an AND gate to be garbled with zero ciphertexts whenever at least one of its input wires is freely settable by the garbler. Our full solution, AuthOr, then combines the ideas from the information-theoretic garbling of Kondi and Patra (Crypto '17) with the improved HG-based scheme we obtain. AuthOr has a lower communication cost (i.e., garbled-circuit size) than HG garbling without any additional security assumptions. Theoretically, AuthOr's garbled-circuit size reduction over HG garbling ranges from 0\% to 100\%, with the exact improvement depending on the circuit structure. We have implemented our scheme and evaluated it on various circuits constructed by independent researchers. Our experimental results show that, in practice, the garbled-circuit size reduction can reach roughly 98\%. 
\keywords{Garbled circuits \and Privacy-free garbling \and Verifiable computing \and Zero-knowledge proofs.}
\end{abstract}
\section{Introduction} \label{Intro}

\textbf{Garbled circuits.} Introduced by Andrew Yao \cite{Yao82}, garbled circuits are a fundamental primitive in cryptography with applications to secure two-party computation \cite{Yao86,Lindell2007,LP11,LR15}, secure multi-party computation \cite{LPSY15,WRK17}, identity-based encryption \cite{DG17}, verifiable outsourced computing \cite{GGP10,MLR23}, and zero-knowledge protocols \cite{JKO13,FNO15,CGM16,KP17,GKPS18}. In its generic form, a garbled-circuit (GC) protocol is a two-party protocol between a garbler and an evaluator. The garbler garbles a Boolean circuit $f$ known to both parties, and the evaluator evaluates the garbled circuit on encoded inputs to obtain encoded outputs. Classically, the garbling procedure assigns two $\ell$-bit keys\footnote{Throughout this work, $\ell\in O(\lambda)$ is the key size based on the security parameter $\lambda$.} $W_i^0$ and $W_i^1$ (for truth values 0 and 1) to each wire $i$ and produces the garbled circuit $\hat{F}$. The evaluator receives $\hat{F}$ together with the key $W_i$ corresponding to the true value $w_i$ of each input wire $i$, then proceeds gate-by-gate in topological order to obtain the output-wire keys. Bellare et al.\ \cite{BHR12} formalized garbled circuits as a cryptographic primitive with three security notions: privacy, obliviousness, and authenticity. Privacy and obliviousness concern hiding the wire values from the evaluator, while authenticity prevents the evaluator from forging encodings of “fake” outputs that do not result from honest evaluation. The efficiency of a garbled circuit is typically measured by the computational cost for the garbler, the computational cost for the evaluator, and the size of the garbled circuit. The last of these is often the main bottleneck \cite{ZRE15}, as it directly determines the communication cost.

\textbf{Privacy-free garbling.} Frederiksen et al.\ \cite{FNO15} showed that authenticity-oriented garbling schemes that drop privacy and obliviousness can be significantly more efficient than full-fledged schemes. The key idea is that when the evaluator already knows the true values $w_a$ and $w_b$ of a gate’s inputs, it can exploit this knowledge during evaluation, reducing overhead. Such schemes are useful for verifiable outsourced computation \cite{GGP10,MLR23} and for efficient zero-knowledge proof protocols \cite{JKO13,FNO15,CGM16,KP17,GKPS18}. The latter application additionally requires a verifiability property so that a malicious verifier cannot construct a garbled circuit whose output keys would reveal the prover’s witness.

There is a quest for achieving an optimal privacy-free garbling scheme. Previous works include the privacy-free schemes by \cite{FNO15}, i.e., GRR2 (compatible with FreeXOR), GRR1, and GRR1 with FleXOR. The most relevant techniques to our work are the privacy-free version of half-gates (HG) garbling \cite{ZRE15} and the information-theoretic (IT) garbling \cite{KP17}. If the circuit is a Boolean formula, i.e., it consists of only fan-out-1 gates, IT garbling is optimal with a garbled circuit of size zero and no cryptographic operations. Although Boolean formulas are already an important class of circuits with zero-knowledge proof-of-satisfiability applications \cite{GGPR13}, this is still quite restrictive. According to the Shannon bound \cite{Sha49,SL23}, the great majority of circuits with $n$ input bits have size $\Theta(2^n/n)$, whereas a Boolean formula with $n$ input bits can have size $O(n)$. For fairness, we converted an example generic circuit into a Boolean formula (see Appendix \ref{inefficiency_of_IT}), which resulted in a circuit of size exponential in the size of the original circuit, implying that both garbling and evaluation have exponential complexity\footnote{The authors also acknowledge this as “not generally efficient for circuits that are not largely formulaic” in \cite{KP17}.}. To the best of our knowledge, before this work, HG garbling has remained the most efficient privacy-free garbling scheme applicable to all Boolean circuits.

\textbf{Our contributions.} We provide a novel privacy-free garbling scheme that is FreeXOR-compatible and has lower AND-gate costs than HG garbling. Our technique benefits from the incorporation of IT garbling and a more efficient version of HG garbling that we also propose here. More concretely, we achieve the following:
\begin{enumerate}
\item In Section \ref{gates}, along with the previously developed garbling schemes that we use in our final solution, we propose an HG-based AND-gate garbling scheme that is FreeXOR-compatible. This scheme allows the garbler to garble an AND gate with zero ciphertexts whenever one of the gate’s input wires has not already been set by another gate that precedes it in topological order.
\item In Section \ref{author}, we combine the scheme mentioned in the previous step with IT garbling to obtain our final privacy-free garbling scheme AuthOr. The garbling algorithm $\mathsf{Gb}$ of AuthOr selects, for each gate, the most suitable garbling technique and benefits from IT garbling and our proposal from the previous step as much as possible, while the evaluation algorithm $\mathsf{Ev}$ detects for each gate which technique was used and proceeds accordingly. Theoretically, for any circuit, compared with HG garbling AuthOr achieves: (i) the same computational security guarantee based on the circular correlation robustness (CCR) of the hash function, (ii) a 0\% to 100\% reduction in garbled-circuit size, and (iii) a 0\% to 100\% reduction in the number of hash-function evaluations for both the garbler and the evaluator. The exact improvement depends on the circuit.
\item In Section \ref{implem}, we provide our implementation results for both our scheme and HG garbling by testing those schemes on the Boolean circuits available from \cite{benchbf,benchbch}. While for the cryptographic circuits our efficiency improvement can be considered marginal, for other circuits we achieve up to a 98.44\% reduction in garbled-circuit size without increasing the computational costs.
\item In Section \ref{sec_proof}, we show the security of our scheme based on the security (i.e., authenticity of \cite{BHR12} and verifiability of \cite{JKO13}) of HG garbling, which in turn is based on the circular correlation robustness of the hash scheme used \cite{ZRE15}.
\end{enumerate}

\begin{table}[t]
\centering
\begin{tabular}{ccccccccccccc} 
\hline
\vspace{-9pt}\\
 {\textbf{Technique}}& & & $\mathsf{Gb}$\textbf{ cost}& & & $\mathsf{Ev}$\textbf{ cost}& & &\textbf{GC size (bits)} & & &\textbf{Assump.} \\
\vspace{-9pt}\\
\hline
\vspace{-8pt}\\
 GRR2+FreeXOR & & &$3g_A s$ & & & $g_A s$ & & &$2g_A \ell $  & & &CCR\\
 \vspace{-8pt}\\
  GRR1 & & &$3g_A s$ & & & $g_A s$  & & &$ g\ell$ & & &PRF \\
  \vspace{-8pt}\\
   GRR1+FleXOR & & &$3g_A s$ & & & $g_A s$ & & &$ [g_A\ell,(g_A+2g_X)\ell]$ & & &CCR\\
   \vspace{-8pt}\\
        Fast garbling & & &$(4g_A+3g_X) s$&  & & $(2g_A+1.5g_X) s$ & & &$ (2g_A+g_X)\ell$& & &PRF\\
     \vspace{-8pt}\\
   Three halves & & &$[3g_A s,6g_A s]$&  & & $[1.5g_A s,3g_A s]$ & & &$ 1.5g_A\ell+5g$& & &CCR\\
     \vspace{-8pt}\\
    IT & & &$[\Omega(g),exp(g)]$& & & $[\Omega(g),exp(g)]$& & &0& & &None\\
    \vspace{-8pt}\\
     HG & & &$2g_A s$& & & $g_A s$ & & &$ g_A\ell$& & &CCR \\
     \vspace{-8pt}\\
    \textbf{AuthOr}& & &$\bm{[0,2g_A s]}$&  & & $\bm{[0,g_A s]}$ & & &$ \bm{[0,g_A\ell]}$& & &\textbf{CCR}\\
    \vspace{-8pt}\\
\hline
\end{tabular}
\caption{ Comparison of state-of-the-art garbling schemes with authenticity, evaluating garbling ($\mathsf{Gb}$), evaluation ($\mathsf{Ev}$), GC size, and security assumptions (CCR=circular correlation robustness and PRF=pseudo-random function). Circuit costs are parameterized by the total number $g$ of gates, the total number $g_A$ of AND gates, the total number $g_X$ of XOR gates, symmetric-key operation cost $s$, and ciphertext size $\ell \in O(\lambda)$. Ranges reflect variations across circuit structures for $f$. All schemes include additional $\Theta(g)$ computation overheads from basic operations.
}
\label{comparison}
%\vspace{-25pt}
\end{table}
\textbf{Comparison with the state-of-the-art.} Table \ref{comparison} compares our scheme with the mentioned privacy-free garbling schemes, the symmetric-key-based full-fledged fast garbling of \cite{GLNP15}, and three-halves of \cite{RR21}. Circular correlation robustness (CCR) \cite{CKKZ12} is a standard assumption in FreeXOR- and FleXOR-compatible schemes, while the security of the non-compatible ones can be based directly on the pseudorandomness of the underlying primitive.

\textbf{Comparison to other proposals.} We would like to mention some other approaches than symmetric-key-based ones for fairness. The gate-ID-based garbling scheme of \cite{KKKS15} uses trapdoor permutations (i.e., public-key operations) to garble a circuit, resulting in a garbled circuit with 4 bits per gate and an additional 2 ciphertexts for AND gates that share input wires with other gates. While in terms of garbled-circuit size $[4g,4g+(g-1)\cdot 2ct]$ the scheme of \cite{KKKS15} improves efficiency compared with previous schemes, it relies on public-key primitives. Other approaches include succinct garbling schemes \cite{CHJV15,AL18} and reusable garbling schemes \cite{GKP+13,BGG+14,Agr17}. However, they induce heavy computational costs and are based on strong assumptions, e.g., the existence of indistinguishability obfuscation / functional encryption / fully homomorphic encryption, the learning-with-errors (LWE) assumption, or assumptions on multilinear maps. While one recent work \cite{LWYY24} based on the LWE assumption has a very low garbled size ($g$ bits), its concrete computational costs are heavy, with a large number of homomorphic operations per gate.

\section{Background on Circuits and Privacy-Free Garbling}
\textbf{Boolean circuits and notation.} A Boolean circuit $f$ is a directed acyclic graph with gates, circuit inputs, and circuit outputs as vertices, and wires as edges. Throughout this work, we use the following notation. $i.\mathsf{GateInputs}$ denotes the input wire(s) of the gate $i$. $f.\mathsf{TGates}$ denotes the set of all gates with type $\mathsf{T}$ (e.g. AND, XOR, NOT) in the circuit $f$. $i.\mathsf{Type}$ denotes the type variable associated with the wire $i$. Each gate $i$ is identified by the index of its output wire $i$.

\textbf{Garbled circuits and notation.} Throughout this work,
\begin{itemize}
\item $a\twoheadleftarrow B$ denotes that $a$ is picked uniformly at random from the set $B$,
\item $a\gets B$ denotes that $a$ is set as an execution of the algorithm $B$,
\item $\ell\in O(\lambda)$ is the key length chosen based on the security parameter and the symmetric-key primitive,
\item $\lambda$ denotes the security parameter. $f.\mathsf{Inputs}$ denotes the input wires of the circuit $f$,
\item $f.\mathsf{Outputs}$ denotes the output wires of the circuit $f$,
\item $\hat{a}$ denotes a straightforward collection of the values produced as $a_i$,
\item $\hat{F}.\mathsf{Next}$ denotes the iterator over the garbled circuit $\hat{F}$.
\end{itemize}

This work follows the garbling-scheme abstraction introduced by \cite{BHR12}. A garbling scheme is composed of the following algorithms:
\begin{itemize}
\item $\mathsf{Gb}$: Takes as input $1^\lambda$ and a Boolean circuit $f$, and outputs $\hat{F}$, $\hat{e}$, and $\hat{d}$, which denote the garbled circuit, encoding information, and decoding information, respectively.
\item $\mathsf{Ev}$: Takes as input $f$, $\hat{F}$, and $\hat{X}$, and outputs the garbled output $\hat{Y}$.
\item $\mathsf{En}$: Takes as input $\hat{e}$ and the plaintext input bit string $\hat{x}$ to $f$, and outputs the garbled input $\hat{X}$.
\item $\mathsf{De}$: Takes as input $\hat{d}$ and $\hat{Y}$, and outputs the plaintext output bit string $\hat{y}$ or $\bot$.
\item $\mathsf{Ve}$: Takes as input $\hat{e}$, $\hat{F}$, and $f$ and outputs a single bit $b$ (1 for verified).
\end{itemize}

\begin{definition}[Correctness] A garbling scheme is correct if for every circuit $f:\{0,1\}^n\rightarrow\{0,1\}^m$ and its input string $\hat{x}\in \{0,1\}^n$, $f(x)=\mathsf{De}\Big(\hat{d},\mathsf{Ev}\big(f,\hat{F},\mathsf{En}(\hat{e},\hat{x})\big)\Big)$ holds
for each execution $(\hat{F},\hat{e},\hat{d})\gets\mathsf{Gb}(1^\lambda, f)$.
\end{definition}

As we are interested in privacy-free garbling schemes, we provide the authenticity definition of \cite{BHR12} below. The definition briefly ensures that no PPT adversary can come up with incorrect output-wire keys that are decodable by $\mathsf{De}$.
\begin{definition}[Authenticity] A garbling scheme is authentic if for every PPT adversary $\mathcal{A}$, every circuit $f:\{0,1\}^n\rightarrow\{0,1\}^m$ and its input string $\hat{x}\in \{0,1\}^n$, there exists a negligible function $\mathsf{n}$ such that:
$$Pr[(\hat{Y}\neq \mathsf{Ev}(f,\hat{F},\hat{X})) \wedge (\mathsf{De}(\hat{d},\hat{Y})\neq \bot ) : (\hat{F},\hat{e},\hat{d})\gets\mathsf{Gb}(1^\lambda, f),$$
$$ \hat{X}\gets\mathsf{En}(\hat{e},\hat{x}), \hat{Y} \gets \mathcal{A}(\hat{F},\hat{X})] \leq\mathsf{n}(\lambda).$$
\end{definition}

Additionally, we require the privacy-free garbling scheme to have the verifiability property of \cite{JKO13}, so that the obtained garbling scheme can easily be plugged into their given zero-knowledge proof-of-knowledge (ZKP) protocol. In this ZKP scheme, the verifier plays the garbler role during the garbled-circuit subprotocol, yet security against a malicious verifier is still achieved thanks to the verifiability of the garbled circuit. We note that the definition below is a slightly modified version of the one given in \cite{JKO13} for arbitrary output lengths by giving $y$ as input to $\mathsf{Ext}$.
\begin{definition}[Verifiability] A garbling scheme is verifiable if for every PPT adversary $\mathcal{A}$, every circuit $f:\{0,1\}^n\rightarrow\{0,1\}^m$, there exists an expected polynomial-time algorithm $\mathsf{Ext}$ and a negligible function $\mathsf{n}$ such that:
\begin{enumerate}
\item $Pr[\mathsf{Ev}(f,\hat{F},\hat{X}_0)\neq\mathsf{Ev}(f,\hat{F},\hat{X}_1),\text{ } 
\mathsf{Ve}(\hat{e},\hat{F}, f)=1, \text{ } f(\hat{x}_0)=f(\hat{x}_1):\text{ } 
(\hat{F},\hat{e},\hat{x}_0,\hat{x}_1)\gets\mathcal{A}(1^\lambda, f),\text{ }  \hat{X}_0\gets\mathsf{En}(\hat{e},\hat{x}_0),\text{ } \hat{X}_1\gets\mathsf{En}(\hat{e},\hat{x}_1)] \leq \mathsf{n}(\lambda)$, and
\item $Pr[\mathsf{Ext}\big(\hat{F},\hat{e},f(\hat{x})\big)\neq\mathsf{Ev}(f,\hat{F},\hat{X}),\text{ }  \mathsf{Ve}(\hat{e},\hat{F}, f)=1: (\hat{F},\hat{e},\hat{x})\gets\mathcal{A}(1^\lambda, f), \hat{X}\gets\mathsf{En}(\hat{e},\hat{x})] \leq \mathsf{n}(\lambda)$.
\end{enumerate}
\end{definition}

\section{Gate Garbling Algorithms and Our Optimization on Half Gates} \label{gates}
In this section, we elaborate on the algorithms that we use in AuthOr, i.e., privacy-free half-gates-based $\mathsf{HG2}$, $\mathsf{HG1}$, $\mathsf{HG0}$, and information-theoretic $\mathsf{ITAND}$ for garbling AND gates; FreeXOR-based $\mathsf{FreeXOR2}$, $\mathsf{FreeXOR1}$, $\mathsf{FreeXOR0}$, and information-theoretic $\mathsf{ITXOR}$ for garbling XOR gates; and $\mathsf{FwNOT}$ and $\mathsf{BwNOT}$ for garbling NOT gates. Among the given gate-garbling techniques, $\mathsf{HG1}$ and $\mathsf{HG0}$ are our contribution. We highlight that all the given algorithms garble gates in the forward direction in topological order, except for the IT-garbling algorithms, which garble in the backward direction from outputs to inputs; we postpone how to determine which algorithm will be used for a given gate to Section \ref{author}. Tables \ref{gate_garbling}, \ref{gate_evaluation}, and \ref{gate_verif} show the formal garbling, evaluation, and verification algorithms, respectively.

\textbf{FreeXOR ($\mathsf{FreeXORz}$) \cite{KS08}.}
The FreeXOR technique \cite{KS08} requires that the keys for each wire $i$ have the same offset $\Delta$, so that $W_i^1=W_i^0\oplus \Delta$ always holds. Given the input-wire keys $W_a^0$ and $W_b^0$, FreeXOR sets the output-wire key $W_i^0\gets W_a^0\oplus W_b^0$. Depending on the number $z$ of inputs that have been set before garbling the XOR gate, this scheme is called $\mathsf{FreeXORz}$ for $z\in \{0,1,2\}$ and the variants differ by picking the unset wire key uniformly at random. The corresponding garbling algorithms are named $\mathsf{GbFreeXOR2}$, $\mathsf{GbFreeXOR1}$, and $\mathsf{GbFreeXOR0}$; the evaluation algorithm for all XOR gates is $\mathsf{EvXOR}$.

\begin{table}[t]
 %\vspace{-15pt}
\centering
\scriptsize
\begin{tabular}{|l|l|l|} 
\hline
\textbf{procedure} $\mathsf{GbHG2}(i,W_a^0,W_b^0,\Delta)$:                                                & \textbf{procedure} $\mathsf{GbHG1}(i,W_a^0,\Delta)$:                                          & \textbf{procedure} $\mathsf{GbHG0}(i,\Delta)$:  \\
\quad $ W_i^0 \gets H(i,W_a^0)$                                                                                              & \quad $ W_i^0 \gets H(i,W_a^0)$                                                    		        &  \quad $W_a^0\twoheadleftarrow \{0,1\}^\ell$ \\
\quad $ F_i\gets H(i,W_a^0)\oplus $                                							& \quad $W_b^0\gets H(i,W_a^0)\oplus $                           					& \quad $ W_i^0 \gets H(i,W_a^0)$  \\
 \quad\quad  \quad $ H(i,W_a^0\oplus\Delta)\oplus W_b^0      $                                                  &     \quad\quad  \quad            $H(i,W_a^0\oplus\Delta)$                                       & \quad $W_b^0\gets H(i,W_a^0)\oplus H(i,W_a^0\oplus\Delta)$  \\
\quad \textbf{return} $ (F_i,W_i^0)$                                                                                             &    \quad \textbf{return} $ (W_b^0,W_i^0)$                                                                &  \quad \textbf{return} $ (W_a^0,W_b^0,W_i^0)$  \\
                                                                                                                                                
                                                                                                                                                   &                                                                                                                                &\\
 \textbf{procedure} $\mathsf{GbFreeXOR1}(W_a^0)$:                                           		    & \textbf{procedure} $\mathsf{GbFreeXOR0}()$:                                                        & \textbf{procedure} $\mathsf{GbFreeXOR2}(W_a^0,W_b^0)$:  \\
\quad $W_b^0\twoheadleftarrow \{0,1\}^\ell$                                                  			    &  \quad $W_a^0\twoheadleftarrow \{0,1\}^\ell$                                                         & \quad $ W_i^0 \gets W_a^0 \oplus W_b^0$   \\
\quad $ W_i^0 \gets W_a^0 \oplus W_b^0$                                                           		    & \quad $W_b^0\twoheadleftarrow \{0,1\}^\ell$                                                         & \quad \textbf{return} $ W_i^0$ \\
 \quad \textbf{return} $ (W_b^0,W_i^0)$                        			                                      & \quad $ W_i^0 \gets W_a^0 \oplus W_b^0$                                                        &\\
		                                                                                                                                  &  \quad \textbf{return} $ (W_a^0,W_b^0,W_i^0)$                                                       &\\
                                                                                                                                                   &                                                                                                                                &\\ 
\textbf{procedure} $\mathsf{GbITAND}(W_i^0,W_i^1)$:                                                           & \textbf{procedure} $\mathsf{GbITXOR}(W_i^0,W_i^1)$:                                       & \textbf{procedure} $\mathsf{GbBwNOT}(W_i^0,W_i^1)$:  \\
\quad $ W_a^0, W_b^0 \gets W_i^0 , W_i^0$                                                                         & \quad $W_a^1\twoheadleftarrow \{0,1\}^\ell$                                                  	&  \quad $ (W_a^0,W_a^1) \gets (W_i^1,W_i^0)$ 		 \\
\quad $ W_a^1\twoheadleftarrow \{0,1\}^\ell$                                                                            & \quad $W_b^1\gets W_a^1 \oplus W_i^0$                                                  		& \quad \textbf{return} $ (W_a^0,W_a^1)$  \\
\quad $ W_b^1\gets W_a^1 \oplus W_i^1$                                                                                & \quad $W_a^0\gets W_b^1 \oplus W_i^1$ 			                                          &  \\
\quad \textbf{return} $ (W_a^0,W_a^1, W_b^0,W_b^1)$ 					                     & \quad $ W_b^0 \gets W_a^1 \oplus W_i^1$                                                          & \textbf{procedure} $\mathsf{GbFwNOT}(W_a^0,\Delta)$: \\
                                                                                                                                                   & \quad \textbf{return} $(W_a^0,W_a^1, W_b^0,W_b^1)$                                        & \quad $ W_i^0 \gets W_a^0 \oplus\Delta$   \\
                                                                                                                                                   &                                                                                                                                & \quad \textbf{return} $ W_i^0$   \\
                                                                                                                                                                                                                                                                                                                                                                                                                                                                                                                       
\hline
\end{tabular}
\caption{Gate garbling algorithms that are used in AuthOr.}
 \label{gate_garbling}
 %\vspace{-25pt}
\end{table}

\textbf{Privacy-free Half-Gates Garbling ($\mathsf{HGz}$).} Here we provide the $\mathsf{HG2}$ of \cite{ZRE15} and our contributions $\mathsf{HG1}$ and $\mathsf{HG0}$ obtained from it.

\textbf{The $\mathsf{HG2}$ algorithm \cite{ZRE15}.} The privacy-free version of the half-gates protocol \cite{ZRE15} is compatible with FreeXOR, i.e., the keys for each wire $w_i$ have the same offset $\Delta$. Here, we provide a slightly modified version such that the hash call $H$ takes as input the id $i$ of the gate instead of the one in \cite{ZRE15}\footnote{In \cite{ZRE15}, the hash function takes as input a value that is incremented with each AND gate garbled. This modification does not affect the security at all, as their proof is based on the uniqueness of this value.}. Given the input-wire keys $W_a^0$, $W_b^0$, and the offset $\Delta$, the garbling algorithm $\mathsf{GbHG2}$ for the gate with index $i$ results in one ciphertext which is set as $F_i\gets H(i,W_a^0)\oplus H(i,W_a^1)\oplus W_b^0$ where $H:\{0,1\}^*\rightarrow \{0,1\}^\ell$ is a hash function with circular correlation robustness \cite{CKKZ12}. The output-wire key for 0 is set as $W_i^0=H(i,W_a^0)$. As this scheme requires that the input keys for both wires have already been fixed before the garbling, we call this gate garbling $\mathsf{HG2}$. We note that among all gate-garbling schemes that we use, this one is the only one that results in a ciphertext that needs to be sent to the evaluator. If the truth value $w_a=0$, the evaluation algorithm $\mathsf{EvHG2}$ sets $W_i\gets H(i,W_a)$. Otherwise, $W_i\gets F_i\oplus H(i,W_a)\oplus W_b$.

\textbf{Our contribution: The $\mathsf{HG1}$ and $\mathsf{HG0}$ algorithms.} If the keys of only one input wire have not been fixed before the garbling of a gate, then the garbler can choose such that $F_i=0^\ell$ and it does not need to be sent anymore. More concretely, given the input-wire key $W_a^0$ and the offset $\Delta$, the key of the output wire $i$ is set as $W_i^0=H(i,W_a^0)$ by $\mathsf{GbHG1}$, and the key for $w_b$ is set as $W_b^0\gets H(i,W_a^0)\oplus H(i,W_a^1)$. If none of the wires were set before the gate being garbled, one wire key is picked randomly as $W_a^0\twoheadleftarrow \{0,1\}^\ell$ and the algorithm $\mathsf{GbHG0}$ continues in the same manner as $\mathsf{HG1}$. Given the input keys $W_a$ and $W_b$, if $w_a=0$, the evaluation algorithm $\mathsf{EvHG0\&1}$ for both $\mathsf{HG1}$ and $\mathsf{HG0}$ sets $W_i\gets H(i,W_a)$. Otherwise, $W_i\gets H(i,W_a)\oplus W_b$.

\begin{table}[t]
\centering
\scriptsize
\begin{tabular}{|l|l|} 
\hline
                                                                                                                                                                                                                                                                                                                                                                                                                                 
\textbf{procedure} $\mathsf{EvHG2}(i,F_i,W_a,W_b,w_a,w_b)$:                                                & \textbf{procedure} $\mathsf{EvHG0\&1}(i,W_a,W_b,w_a,w_b)$:                         \\ 
\quad \textbf{if} $w_a=0$ \textbf{then}                                                                                       & \quad \textbf{if} $w_a=0$ \textbf{then}                                                                   \\
\quad \quad $W_i\gets H(i,W_a)$                                                                                              & \quad \quad $W_i\gets H(i,W_a)$                                                                          \\
\quad \textbf{else}                                                                                                                      & \quad \textbf{else}                                                                                                    \\
\quad \quad $W_i\gets H(i,W_a)\oplus W_b \oplus F_i$                                                              & \quad \quad $W_i\gets H(i,W_a)\oplus W_b$                                                         \\
\quad \textbf{return} $ W_i$                                                                                                       & \quad \textbf{return} $ W_i$                                                                                   \\                                                                                                                                                   
                                                                                                                                                   &                                                                                                                                \\

  \textbf{procedure} $\mathsf{EvXOR}(W_a,W_b,w_a,w_b)$:                                                                                            &\textbf{procedure} $\mathsf{EvITAND}(W_a,W_b,w_a,w_b)$:    \\
\quad $ W_i \gets W_a \oplus W_b$                                                                                                                    	& \quad \textbf{if} $w_a=0$  \textbf{then}   \\
\quad \textbf{return} $ W_i$                                                                                                                       		& \quad \quad$ W_i\gets W_a$ \\
                                                                             		                                          & \quad \textbf{else if} $w_b=0$ \textbf{then} \\
	\textbf{procedure} $\mathsf{EvNOT}(W_a)$:				                                                                    & \quad \quad$ W_i\gets W_b$ \\
            \quad $W_i\gets W_a$                                                                                                                                              & \quad \textbf{else} \\
                  \quad          \textbf{return} $ W_i$                                                                                                                                                                                                                               & \quad \quad$ W_i\gets W_a\oplus W_b$\\
                                                                                                                                                                                                                                                                            &  \quad \textbf{return} $ W_i$ \\                                                                                                                                                                                                                                                                                                                                                                                                          
\hline
\end{tabular}
\caption{Gate evaluation algorithms that are used in AuthOr.}
 \label{gate_evaluation}
  %\vspace{-15pt}
\end{table}

\begin{table}[t]
\centering
\scriptsize
\begin{tabular}{|l|l|} 
\hline
\textbf{procedure} $\mathsf{VeHG2}(i,W_a^0,W_a^1,W_b^0,W_b^1,F_i)$:			&    \textbf{procedure} $\mathsf{VeHG0\&1}(i,W_a^0,W_a^1,W_b^0,W_b^1)$:		 \\
\quad \textbf{if} $W_a^0\oplus W_a^1\neq W_b^0\oplus W_b^1$					&\quad \textbf{if} $W_a^0\oplus W_a^1\neq W_b^0\oplus W_b^1$ 				\\
 \quad \quad \textbf{return} $0$												& \quad \quad \textbf{return} $0$										 \\
 \quad $\Delta\gets W_a^0\oplus W_a^1$										& \quad $\Delta\gets W_a^0\oplus W_a^1$								\\
  \quad $K_{1}\gets H(i,W_a^0)$                                                               				&  \quad $K_{1}\gets H(i,W_a^0)$  									 \\        
   \quad $K_{2}\gets W_b^0\oplus H(i,W_a^0\oplus\Delta) \oplus F_i$                                     &    \quad $K_{2}\gets W_b^0\oplus H(i,W_a^0\oplus\Delta) $    				\\ 
      \quad \textbf{if} $K_{1}\neq K_{2}$                                                                 			& \quad \textbf{if} $K_{1}\neq K_{2}$     									 \\   
       \quad \quad \textbf{return} $0$ 					                                                   & \quad \quad \textbf{return} $0$ 		 			\\      
            \quad  \textbf{return} $(K_{1},K_{1}\oplus\Delta)$                                                       &      \quad  \textbf{return} $(K_{1},K_{1}\oplus\Delta)$        													\\       
                   																		    &                                                                                                                                 \\ 
\textbf{procedure} $\mathsf{VeXOR}(W_a^0,W_a^1,W_b^0,W_b^1)$:		    &  \textbf{procedure} $\mathsf{VeNOT}(W_a^0,W_a^1)$:  \\ 
\quad \textbf{if} $W_a^0\oplus W_a^1\neq W_b^0\oplus W_b^1$  									                                  & \quad \textbf{return} $(W_a^1,W_a^0)$  \\         
  \quad \quad \textbf{return} $0$    																	&                       \\        
\quad  \textbf{return} $(W_a^0\oplus W_b^0,W_a^0\oplus W_b^1) $  																           &           \textbf{procedure} $\mathsf{VeITAND}(W_a^0,W_a^1,W_b^0,W_b^1):	$             		\\   
						      											&                      \quad  \textbf{if} $W_a^0\neq W_b^0$                \\   
						      											&                       \quad \quad \textbf{return} $0$             \\ 
						      											&                       \quad  \textbf{return} $(W_a^0,W_a^1\oplus W_b^1)$           \\ 																																		                                                                                                                                                                                                                                                                                                                                                                                                                                                                                                                              
\hline
\end{tabular}
\caption{Gate verification algorithms that are used in AuthOr.}
 \label{gate_verif}
  %\vspace{-30pt}
\end{table}

\textbf{Information-Theoretic (IT) Garbling \cite{KP17}.} Kondi and Patra \cite{KP17} provide an elegant way of privacy-free garbling of AND and XOR gates that does not require costly hash-function calls and results in size-zero garbled circuits. However, if during the evaluation of an AND gate $w_a=0$ and $w_b=1$, their scheme leaks both keys $W_b^0$ and $W_b^1$ as it sets $W_a^0=W_b^0$ during garbling. Yet, the authors prove in their work that this leakage is not a problem (as the output-wire key is never leaked) as long as there are no fan-out-2 gates in the circuit. Due to this leakage, their scheme does not have a global offset $\Delta$. Nevertheless, the authors still show that XOR gates can be garbled with a local offset “freely”. The garbling is “backward”, i.e., gates are visited in inverse topological order, so that when a gate is being garbled, either its output keys are already set or the garbler picks them randomly as $W_i^0\twoheadleftarrow \{0,1\}^\ell$ and $W_i^1\twoheadleftarrow \{0,1\}^\ell$.

\textbf{The $\mathsf{ITAND}$ algorithm.} Given the output keys $W_i^0$ and $W_i^1$, the garbling algorithm $\mathsf{GbITAND}$ sets the keys for the truth value zero for both input wires as $W_a^0\gets W_i^0$ and $W_b^0\gets W_i^0$. Next, it picks $W_a^1\twoheadleftarrow \{0,1\}^\ell$ and computes $W_b^1\gets W_a^1\oplus W_i^1$. During evaluation, $\mathsf{EvITAND}$ sets $W_i\gets W_a$ if $w_a=0$, $W_i\gets W_b$ if $w_a=1$ and $w_b=0$, and $W_i\gets W_a\oplus W_b$ otherwise.

\textbf{The $\mathsf{ITXOR}$ algorithm.} Given the output keys $W_i^0$ and $W_i^1$, the garbling algorithm $\mathsf{GbITXOR}$ sets the local offset for the XOR gate being garbled as follows. It picks $W_a^1\twoheadleftarrow \{0,1\}^\ell$ and computes $W_b^1\gets W_a^1\oplus W_i^0$, $W_a^0\gets W_b^1\oplus W_i^1$, and $W_b^0\gets W_a^1\oplus W_i^1$. During evaluation, $\mathsf{EvXOR}$ sets $W_i\gets W_a\oplus W_b$.

\textbf{Garbling NOT Gates \cite{Kol05,KP17}.} In both the forward and backward directions, NOT gates can be freely garbled by simply swapping the association of keys for 0 and 1 as in \cite{Kol05,KP17}. That is, given a NOT gate with input wire $a$ and output wire $i$, $W_a^0=W_i^1$ and $W_a^1=W_i^0$; based on the direction in which this is done, we name the algorithms $\mathsf{GbFwNOT}$ for the forward direction and $\mathsf{GbBwNOT}$ for the backward direction. The evaluation algorithm $\mathsf{EvNOT}$ simply copies the key from the input to the output in the forward direction, i.e., the key does not change while the truth value changes.

\textbf{Verification Algorithms.} Table \ref{gate_verif} shows the verification algorithms for the provided gate-garbling algorithms. Given both keys for each input wire, they check whether a gate is garbled correctly. If so, they return 1 and the keys for the output wire; otherwise, they return 0. We have obtained $\mathsf{VeITAND}$ from \cite{KP17} as the verification algorithm for the IT-garbled AND gates, which checks that the output key for 0 that could be found in all cases is equal. The verification algorithm $\mathsf{VeXOR}$ for XOR gates executes based on whether both input wires have the same offset. We constructed the half-gates AND verification algorithms $\mathsf{VeHG2}$ and $\mathsf{VeHG0\&1}$ based on the check that the output key for 0 that could be found in all cases is equal, similar to $\mathsf{VeITAND}$. We highlight that $\mathsf{VeXOR}$, $\mathsf{VeHG2}$, and $\mathsf{VeHG0\&1}$ also check whether both input wires have the same offset.

\section{Our Privacy-Free Garbling Protocol} \label{author}

We provide our AuthOr scheme in Table \ref{our_garbling} that makes use of the gate-garbling algorithms in Section \ref{gates}. We assume both input wires for each AND or XOR gate are distinct, as otherwise our scheme is not immune to the attack mentioned in \cite{NS23}\footnote{If this is not the case, we refer to their work for the possible countermeasures.}.

\begin{table}
\centering
\scriptsize
\begin{tabular}{|l|l|} 
\hline
\textbf{procedure} $\mathsf{Gb}(1^\lambda , f)$:                                                                                                                                              & \textbf{procedure} $\mathsf{Ev}(f,\hat{F},\hat{X})$:  \\
\quad $\Delta \twoheadleftarrow \{0,1\}^\ell$  s.t. $\ell\in O(\lambda)$                                                                                                               	& \quad label each wire as $\mathsf{TypeM}$ or $\mathsf{TypeS}$         \\
\quad label each wire as $\mathsf{TypeM}$ or $\mathsf{TypeS}$                                                                                                                                                                 	& \quad  \textbf{for} each gate $i \in f \text{ }\{in\text{ } topo.\text{ } order\}$ \textbf{do}  \\
\quad $\{Forward$ $phase\}$                                                                                                                                      	&  \quad \quad   \textbf{if} $i \in f.\mathsf{NOTGates}$ \textbf{then}\\
\quad  \textbf{for} each gate $i \in f \text{ }\{in\text{ } topo.\text{ } order\}$ \textbf{do}                                                                           	&  \quad \quad \quad $a \gets i.\mathsf{GateInputs}$   \\
\quad \quad     \textbf{if} $i \in f.\mathsf{NOTGates}$ \textbf{then}													& \quad \quad \quad $W_i \gets \mathsf{EvNOT}(W_a)$   \\
\quad \quad \quad $a \gets i.\mathsf{GateInputs}$																&   \\
 \quad \quad \quad \textbf{if}  $a.\mathsf{Type}=\mathsf{TypeM}$  \textbf{then}											& \quad \quad \textbf{else if} $i \in f.\mathsf{XORGates}$ \textbf{then}  \\
  \quad \quad \quad \quad $W_a^0 \twoheadleftarrow \{0,1\}^\ell$													&  \quad \quad \quad $\{a, b\} \gets i.\mathsf{GateInputs}$    \\
  \quad \quad \quad \textbf{if }   $a.\mathsf{Type}\neq\mathsf{TypeS}$ \textbf{then} 										&  \quad \quad \quad $W_i \gets \mathsf{EvXOR}(W_a, W_b)$    \\
 \quad \quad \quad \quad $W_i^0 \gets \mathsf{GbFwNOT}(W_a^0)$												& \\
																									&  \quad \quad   \textbf{else} \\
\quad \quad     \textbf{else if} $i \in f.\mathsf{XORGates}$ \textbf{then} 												&  \quad \quad \quad $\{a, b\} \gets i.\mathsf{GateInputs}$    \\
\quad \quad \quad     $\{a, b\} \gets i.\mathsf{GateInputs}$ 														& \quad \quad \quad \textbf{if}  $a.\mathsf{Type}=\mathsf{TypeF}$ \textbf{and}  $b.\mathsf{Type}=\mathsf{TypeF}$ \textbf{then} \\
 \quad \quad \quad \textbf{if}  $a.\mathsf{Type}=\mathsf{TypeF}$ \textbf{and}  $b.\mathsf{Type}=\mathsf{TypeF}$ \textbf{then}     	&  \quad \quad \quad \quad $W_i\gets \mathsf{EvHG2}(i,\hat{F}.\mathsf{Next},W_a,W_b,w_a,w_b) $     \\
\quad \quad \quad \quad $W_i^0 \gets \mathsf{GbFreeXOR2}(W_a^0,W_b^0)$ 										& \quad \quad \quad \textbf{else if} $\exists p \in \{a,b \} \text{ } s.t. \text{ }    p.\mathsf{Type}=\mathsf{TypeF}$  \textbf{then} \\
																									&  \quad \quad \quad \quad $W_i \gets \mathsf{EvHG0\&1}(i,W_p,W_{\{a,b\}\setminus p},w_p,w_{\{a,b\}\setminus p})$   \\
\quad \quad \quad \textbf{else if}  $\exists p \in \{a,b \} \text{ } s.t. \text{ } p.\mathsf{Type}=\mathsf{TypeF} $ \textbf{then} 				 & \quad \quad \quad \textbf{else if} $\exists p \in \{a,b \} \text{ } s.t. \text{ } p.\mathsf{Type}=\mathsf{TypeM} $  \textbf{then}        \\
\quad \quad \quad \quad $(W_{\{a,b\}\setminus p} ^0, W_i^0) \gets \mathsf{GbFreeXOR1}(W_p^0)$						  	&  \quad \quad \quad \quad $W_i \gets \mathsf{EvHG0\&1}(i,W_a,W_b,w_a,w_b)$   \\
																							                  &  \quad \quad \quad \textbf{else}                   \\
\quad \quad \quad \textbf{else if} $ \exists p \in \{a,b \}  \text{ } s.t. \text{ }  p.\mathsf{Type}=\mathsf{TypeM}$ \textbf{then}   				& \quad \quad \quad \quad $W_i \gets \mathsf{EvITAND}(W_a,W_b,w_a,w_b)$   \\
\quad \quad \quad \quad $(W_a^0,W_b^0,W_i^0) \gets \mathsf{GbFreeXOR0}()$                                                                                 &    \\
																									&  \quad \quad \textbf{if} gate $i $ has $\mathsf{TypeM}$ or $\mathsf{TypeF}$ input wire  \textbf{then}\\
\quad \quad \textbf{else}    																				&   \quad \quad \quad \textbf{for} each $p\in\{a,b,i\}$ s.t. $ğ.\mathsf{Type}\neq\mathsf{TypeF}$ \textbf{do}\\
\quad \quad \quad     $\{a, b\} \gets i.\mathsf{GateInputs}$															&   \quad \quad \quad \quad  label   $p$ as $\mathsf{TypeF}$	\\
\quad \quad \quad \textbf{if}  $a.\mathsf{Type}=\mathsf{TypeF}$ \textbf{and} $b.\mathsf{Type}=\mathsf{TypeF}$  \textbf{then}         &         \\ 
\quad \quad \quad \quad $(\hat{F}.\mathsf{Next},W_i^0) \gets \mathsf{GbHG2}(i,W_a^0,W_b^0,\Delta)$                                     		&     \quad \textbf{for} $i \in \hat{f}.\mathsf{Outputs}$ \textbf{do}    \\    
																									&    \quad\quad $Y_i \gets W_i$ \\ 
\quad \quad \quad \textbf{else if} $ \exists p \in \{a,b \} \text{ } s.t. \text{ }  p.\mathsf{Type}=\mathsf{TypeF}$  \textbf{then}                                 &   \quad     \textbf{return} $\hat{Y}$       \\
\quad \quad \quad \quad $(W_{\{a,b\}\setminus p} ^0, W_i^0) \gets \mathsf{GbHG1}(i,W_p^0,\Delta)$     		                                    &  \\
																					                                  &  \textbf{procedure} $\mathsf{En}(\hat{e},\hat{x})$: \\
\quad \quad \quad \textbf{else if} $ \exists p \in \{a,b \}  \text{ } s.t. \text{ } p.\mathsf{Type}=\mathsf{TypeM} $  \textbf{then}  				     &     \quad parse $(e_i^0,e_i^1)\gets e_i$  \\                    
\quad \quad \quad \quad $(W_a^0, W_b^0, W_i^0) \gets \mathsf{GbHG0}(i,\Delta)$       									  & \quad $X_i \gets e_i^{x_i} $ \\
																									&    \quad \textbf{return} $\hat{X}$\\
 \quad \quad \textbf{if} gate $i $ has $\mathsf{TypeM}$ or $\mathsf{TypeF}$ input wire \textbf{then}							&       \\
  \quad \quad \quad \textbf{for} each $p\in\{a,b,i\}$ s.t. $p.\mathsf{Type}\neq\mathsf{TypeF}$ \textbf{do}							&   \textbf{procedure} $\mathsf{De}(\hat{d},\hat{Y})$:  \\ 
 \quad \quad \quad \quad $W_p^1\gets W_p^0\oplus \Delta$ 														& \quad  \textbf{for} $d_i\in \hat{d}$ \textbf{do}\\ 
\quad \quad \quad \quad  label   $p$ as $\mathsf{TypeF}$								   								&  \quad \quad parse $(h_0, h_1) \gets d_i$\\ 
 																									&    \quad \quad parse $(h_0, h_1) \gets d_i$ \\      
 \quad $\{Backward$ $phase\}$																			    & \quad \quad \textbf{if}  $H(Y_i,i)=h_0$ \textbf{then} $y_i\gets 0$  \\
 \quad  \textbf{for} each gate $i \in f \text{ } \{in\text{ } inv.\text{ }  topo.\text{ } order\}$ \textbf{do}                                                             &    \quad \quad \textbf{else if}  $H(Y_i,i)=h_1$ \textbf{then} $y_i\gets 1$  \\ 
\quad \quad   \textbf{if} $i.\mathsf{Type}=\mathsf{TypeS}$     \textbf{then}												 &    \quad \quad \textbf{else}  \textbf{return} $\bot$ \\
\quad \quad \quad   $(W_i^0,W_i^1)  \twoheadleftarrow \{0,1\}^\ell\times \{0,1\}^\ell$                                                                                 &  \quad  \textbf{return} $\hat{y}$  \\
																									&  \\
 \quad  \quad \textbf{if} $i \in f.\mathsf{NOTGates}$ \textbf{then}													    &  \textbf{procedure} $\mathsf{Ve}(\hat{e},\hat{F},f)$:   \\  
  \quad \quad  \quad $a \gets i.\mathsf{GateInputs}$     		                 											     &  \quad find which garbling is used in each gate (as in $\mathsf{Ev}$)\\                               
  \quad \quad \quad \textbf{if} $a.\mathsf{Type}=\mathsf{TypeS}$	                                                                                 			  &   \quad  \textbf{for} each gate $i \in f \text{ }\{in\text{ } topo.\text{ } order\}$ \textbf{do} \\
   \quad\quad\quad\quad $(W_a^0,W_a^1)\gets \mathsf{GbBwNOT}(W_i^0,W_i^1)	$									 &\quad\quad run the related gate verification algorithm \\ 
																									 &   \quad\quad \textbf{if} it the gate verification returns $0$ \textbf{then} \\
   \quad  \quad \textbf{if} $i \in f.\mathsf{XORGates}$ \textbf{then} 				            								  &    \quad\quad\quad\textbf{return}  $0$ \\
   	       \quad \quad  \quad $(a,b) \gets i.\mathsf{GateInputs}$													   &  \quad\quad \textbf{else}  \\
       \quad \quad \quad \textbf{if} $a.\mathsf{Type}=\mathsf{TypeS}$  \textbf{then}                            							 &  \quad\quad\quad assign what verification returned as $(W_i^0, W_i^1)$ \\
     \quad \quad  \quad \quad$(W_a^0,W_a^1,W_b^0,W_b^1)\gets \mathsf{GbITXOR}(W_i^0,W_i^1)$  						  &   \quad \textbf{return}  $1$   \\
        																									 &   \\
          \quad  \quad \textbf{else} 																			  &  \\
              \quad \quad \quad $(a,b) \gets \mathsf{GateInputs}(f, i)$ 													 &   \\                          
	\quad \quad \quad \textbf{if} $a.\mathsf{Type}=\mathsf{TypeS}$  	\textbf{then}									  &    \\
     \quad \quad  \quad \quad $(W_a^0,W_a^1,W_b^0,W_b^1)\gets \mathsf{GbITAND}(W_i^0,W_i^1)$ 						&  \\
																									&  \\
 \quad \quad  label all $\mathsf{TypeS}$ input wires of gate $i$ as $\mathsf{TypeB}$ 												& \\
																									 &   \\                
 \quad $\{finalizing\text{ }and \text{ }return \}$                                                                                                                                              &   \\
\quad \textbf{for} $i \in f.\mathsf{Inputs}$ \textbf{do}  															  & \\
 \quad \quad  $e_i \gets (W_i^0,W_i^1)$ 																		  &\\
\quad \textbf{for} $i \in f.\mathsf{Outputs}$  \textbf{do}															  &  \\
 \quad \quad  $d_i \gets (H(W_i^0, i), H(W_i^1, i))$ 																  &\\
 \quad \textbf{return}  $(\hat{F}, \hat{e},\hat{d})$      																&\\
\hline
\end{tabular}
\caption{Our  privacy-free garbling scheme AuthOr.}
 \label{our_garbling}
\end{table}

\textbf{Initialization of our $\mathsf{Gb}$ algorithm.} The $\mathsf{Gb}$ algorithm starts by choosing a global offset $\Delta$. The type of each wire is set by $\mathsf{Gb}$ deterministically and derived from the circuit topology. During the initialization, each wire is set as follows:
\begin{itemize} 
\item $\mathsf{TypeS}$, if that wire is input to at most a single gate.
\item $\mathsf{TypeM}$, otherwise.
\end{itemize}

Here, we do not propose a method for this labeling, as there are simple ways to do this in $O(g)$ that are derived from generic topological-sorting algorithms. Note that the wire labels will be updated again deterministically during the forward and backward phases of the garbling algorithm, which is discussed below.

\textbf{Forward and backward phases of $\mathsf{Gb}$.} In our solution, we attempt to use IT garbling for as many gates as possible. To solve the wire-key leakage issue of IT garbling mentioned in Section \ref{gates}, our garbling algorithm $\mathsf{Gb}$ has two phases: a forward phase, i.e., garbling in topological order with $\mathsf{GbFwNOT}$, $\mathsf{GbHGz}$, and $\mathsf{GbFreeXORz}$ using a global offset $\Delta$; and a backward phase, i.e., garbling in inverse topological order with $\mathsf{GbBwNOT}$, $\mathsf{ITAND}$, and $\mathsf{ITXOR}$.

\textbf{Determination of the gate-garbling algorithms.} The gate-garbling algorithm chosen for a gate is based on the types of its input wires. During the forward phase, only the gates that would result in a problem for IT garbling (if they were left to the backward phase) are garbled, and the offsets of the input and output wires of those gates are set to $\Delta$. Gates with a $\mathsf{TypeM}$ input wire cannot be left to the backward phase (and should be garbled during the forward phase), as IT garbling results in the leakage of both keys for a wire, which is known to be a problem if that wire is input to more than one gate. Also, during the forward phase of $\mathsf{Gb}$, when the keys of a wire are determined with offset $\Delta$, we cannot leave that wire to the backward phase, as the leakage of both keys means leakage of $\Delta$ and the compromise of the security of all wires with the offset $\Delta$. Hence, during the forward phase $\mathsf{Gb}$ replaces such wire types with $\mathsf{TypeF}$ and also garbles each gate that those wires are input to. $\mathsf{Gb}$ leaves for the backward phase only the gates whose both inputs are $\mathsf{TypeS}$ wires, when the time of garbling that gate comes. We highlight that the keys for these wires cannot be set later during the forward phase, due to the following theorem:

\begin{theorem} \label{type_sing}
Once a wire that is input to a gate $i$ is labeled as $\mathsf{TypeS}$ by the $\mathsf{Gb}$ algorithm, then during the forward phase its label may be changed into $\mathsf{TypeF}$ only when a gate $j\leq i$ is being garbled.
\end{theorem}
\begin{proof}
Wires are initially labeled as either $\mathsf{TypeS}$ or $\mathsf{TypeM}$, and the algorithm never reassigns a wire to $\mathsf{TypeS}$ later. Similarly, a wire labeled $\mathsf{TypeS}$ cannot become $\mathsf{TypeM}$, since $\mathsf{TypeM}$ is only assigned at initialization.

A $\mathsf{TypeS}$ wire becomes $\mathsf{TypeF}$ only when it is involved in the garbling of a gate, either as an input or as an output. If $a$ is input to multiple gates, it cannot be labeled $\mathsf{TypeS}$ initially. If it is an output of some gate $j > i$, this violates the topological ordering (since $a$ is also an input to gate $i$), and so such a gate cannot exist. Hence, only gates $j \leq i$ can change the label of $a$ from $\mathsf{TypeS}$ to $\mathsf{TypeF}$.
\end{proof}

Somewhat surprisingly, even the gates whose input wires are $\mathsf{TypeS}$ but whose output wires are converted to $\mathsf{TypeF}$ by garbling another gate afterward can be left to the backward phase, as the offset of an output wire is perfectly hidden in IT garbling. We note that an AND gate that is garbled during the forward phase may still benefit from the efficiency improvements due to our $\mathsf{HG0}$ or $\mathsf{HG1}$ if at least one of its inputs is not $\mathsf{TypeF}$ before garbling this gate. Also, the garbled circuit $\hat{F}$ is determined during this phase, leaving only the complete determination of input and output keys to the backward phase. During the backward phase, all the gates that were left from the forward phase are garbled via IT garbling. Note that these are the gates whose both input wires are $\mathsf{TypeS}$, so their keys can be set by $\mathsf{GbITAND}$ and $\mathsf{GbITXOR}$ without any problem. At the end, all wires become either $\mathsf{TypeF}$ (each with global offset $\Delta$) or $\mathsf{TypeB}$ (each with arbitrary offset). We highlight that if the given circuit $f$ is a Boolean formula, all the gates are garbled in the backward phase, and our scheme is equivalent to IT garbling.

Figure \ref{fig:exampleGarbled} compares the number of $\mathsf{HG2}$ runs employed in garbling a generic circuit for both the HG garbling approach \cite{ZRE15} and AuthOr. In the HG approach, all 4 gates of the given circuit are garbled using $\mathsf{HG2}$, resulting in 4 ciphertexts. In contrast, the AuthOr garbling algorithm applies only one $\mathsf{HG2}$ execution, leading to the generation of a single ciphertext.

\begin{figure}[t]
    \centering
    \begin{subfigure}{0.40\textwidth}
        \centering
        \includegraphics[width=\linewidth]{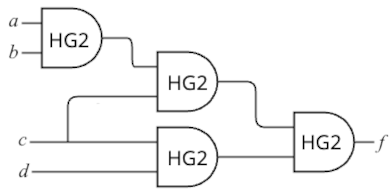} % Replace with your image file
        \caption{HG Garbled}
        \label{fig:hggarbled}
    \end{subfigure}
    \hspace{30pt} % Adds horizontal space between the images
    \begin{subfigure}{0.40\textwidth}
        \centering
        \includegraphics[width=\linewidth]{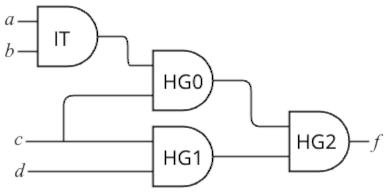} % Replace with your image file
        \caption{AuthOr Garbled}
        \label{fig:authorGarbled}
    \end{subfigure}
    \caption{The impact of using $\mathsf{HG0}$, $\mathsf{HG1}$ and $\mathsf{ITAND}$ on the garbled circuit. The circuits are in topological order, left to right .}
    \label{fig:exampleGarbled}
    %\vspace{-15pt}
\end{figure}

\textbf{Our $\mathsf{Ev}$, $\mathsf{En}$, $\mathsf{De}$, and $\mathsf{Ve}$ algorithms.} The evaluation algorithm $\mathsf{Ev}$ simply runs in topological order and applies the related evaluation algorithms given in Table \ref{gate_evaluation}. What this algorithm needs is only information on which gate-garbling scheme is used for each gate, which is determined using the same method $\mathsf{Gb}$ uses. All the gates garbled in the forward phase of $\mathsf{Gb}$ can be evaluated with the corresponding evaluation algorithm. The encoding algorithm $\mathsf{En}$ simply collects the keys for the circuit input wires as in \cite{ZRE15}. The decoding algorithm $\mathsf{De}$ collects the hash outputs of the keys for the circuit output wires as in \cite{ZRE15}, where the hash algorithm ensures that even with the decoding data, the evaluator is prevented from finding both truth-value keys for the output wire. The verification algorithm $\mathsf{Ve}$ finds the gate-garbling algorithm that is used in each gate and then, in topological order, calls the gate verification algorithms in Table \ref{gate_verif}. We highlight that there is no need for checking that the same global offset $\Delta$ is used in each gate. Instead, we check whether the incoming wires of a gate have the same offset, which is sufficient to ensure that any set of forward-garbled connected\footnote{We slightly abuse the terminology as two gates can be connected via sharing any wires with each other or they are both connected to another gate.} gates have wires with the same local offset. Two such unconnected forward-garbled gate sets could have separate local offsets (although for simplicity, our algorithm does not allow it) and still pass the verification without sacrificing verifiability.

\begin{theorem}
The AuthOr garbling scheme is a privacy-free garbling scheme with correctness, authenticity, and verifiability properties if the hash function $H$ has circular correlation robustness (CCR).
\end{theorem}
We show the security of our scheme in Appendix \ref{sec_proof}.

\begin{table}[h!]
\scriptsize
\centering
\begin{tabular}{|c|ccc|ccc|c|} 
\hline
& \multicolumn{3}{c|}{Half Gates \cite{ZRE15}}& &\textbf{AuthOr} & &  \\ 
\cline{2-8}
{Circuit}&{$\mathsf{Gb}$ cost} & {$\mathsf{Ev}$ cost} & {GC Size} & {$\mathsf{Gb}$ cost} & {$\mathsf{Ev}$ Cost} & {GC Size} & {GC Size Reduction (\%)} \\ 
\hline\hline
adder64 & 0.003 & 0.001 & 64 & 0.010 & 0.001 & 64 & 0.00 \\ \hline
sub64 & 0.002 & 0.001 & 64 & 0.003 & 0.001 & 64 & 0.00 \\ \hline
mult64 & 0.111 & 0.041 & 4034 & 0.116 & 0.042 & 3970 & 1.59 \\ \hline
mult2-64 & 0.224 & 0.085 & 8129 & 0.234 & 0.087 & 4034 & \textbf{50.38} \\ \hline
divide64 & 0.158 & 0.068 & 4665 & 0.169 & 0.070 & 4664 & 0.02 \\ \hline
udivide64 & 0.123 & 0.049 & 4286 & 0.128 & 0.051 & 4225 & 1.42 \\ \hline
and8 & 0.000 & 0.000 & 8 & 0.000 & 0.000 & 1 & \textbf{87.50} \\ \hline
zero-equal & 0.002 & 0.000 & 64 & 0.000 & 0.000 & 1 & \textbf{98.44} \\ \hline
FP-add & 0.140 & 0.049 & 5369 & 0.147 & 0.051 & 5348 & 0.39 \\ \hline
FP-floor & 0.016 & 0.006 & 645 & 0.017 & 0.006 & 644 & 0.16 \\ \hline
FP-ceil & 0.016 & 0.006 & 645 & 0.017 & 0.006 & 645 & 0.00 \\ \hline
FP-div & 2.090 & 0.746 & 80184 & 2.179 & 0.764 & 80163 & 0.03 \\ \hline
FP-sqrt & 2.350 & 0.844 & 89964 & 2.455 & 0.857 & 89964 & 0.00 \\ \hline
FP-f2i & 0.038 & 0.013 & 1464 & 0.042 & 0.014 & 1456 & 0.55 \\ \hline
FP-mul & 0.501 & 0.176 & 19426 & 0.523 & 0.180 & 19406 & 0.10 \\ \hline
FP-eq & 0.009 & 0.003 & 316 & 0.009 & 0.003 & 305 & 3.48 \\ \hline
FP-i2f & 0.064 & 0.022 & 2407 & 0.068 & 0.023 & 2407 & 0.00 \\ \hline
cmp32-unsigned-lt & 0.004 & 0.001 & 151 & 0.004 & 0.001 & 141 & \textbf{6.62} \\ \hline
cmp32-signed-lteq & 0.004 & 0.001 & 151 & 0.004 & 0.001 & 138 & \textbf{8.61} \\ \hline
 & & & & & & & Average: 13.67 \\ \hline 
 \multicolumn{8}{c}{\footnotesize \textbf{(a)}}\\
 \multicolumn{8}{c}{}\\
 \hline
&   \multicolumn{3}{c|}{Half Gates \cite{ZRE15}} & &\textbf{AuthOr} & &  \\ 
\cline{2-8}
{Circuit}&{$\mathsf{Gb}$ cost} & {$\mathsf{Ev}$ cost} & {GC Size} & {$\mathsf{Gb}$ cost} & {$\mathsf{Ev}$ Cost} & {GC Size} & {GC Size Reduction  (\%)} \\ 
\hline \hline
sha256 & 0.767 & 0.325 & 22574 & 0.823 & 0.334 & 22569 & 0.02 \\ \hline
aes128 & 0.210 & 0.087 & 6401 & 0.227 & 0.090 & 6401 & 0.00 \\ \hline
aes192 & 0.238 & 0.098 & 7169 & 0.256 & 0.103 & 7169 & 0.00 \\ \hline
Keccak-f & 1.211 & 0.474 & 38401 & 1.299 & 0.493 & 38401 & 0.00 \\ \hline
sha512 & 1.973 & 0.836 & 57948 & 2.133 & 0.857 & 57943 & 0.01 \\ \hline
aes256 & 0.292 & 0.121 & 8833 & 0.313 & 0.126 & 8833 & 0.00 \\ \hline
 & & & & & & & Average: 0.01 \\ \hline
 \multicolumn{8}{c}{\footnotesize \textbf{(b)}}\\
  \multicolumn{8}{c}{}\\

\hline
& \multicolumn{3}{c|}{Half Gates \cite{ZRE15}} & &\textbf{AuthOr} & &  \\ 
\cline{2-8}
{Circuit}&{$\mathsf{Gb}$ cost} & {$\mathsf{Ev}$ cost} & {GC Size} & {$\mathsf{Gb}$ cost} & {$\mathsf{Ev}$ Cost} & {GC Size} & {GC Size Reduction  (\%)} \\ 
\hline\hline
c5315 & 0.055 & 0.016 & 2012 & 0.044 & 0.014 & 1229 & \textbf{ 38.92 }\\ \hline
c1908 & 0.017 & 0.005 & 576 & 0.014 & 0.004 & 398 &  \textbf{30.90}\\ \hline
c432 & 0.004 & 0.001 & 156 & 0.003 & 0.001 & 83 & \textbf{ 46.79 }\\ \hline
c6288 & 0.065 & 0.020 & 2385 & 0.056 & 0.017 & 1650 & \textbf{ 30.82}\\ \hline
c7552 & 0.070 & 0.021 & 2491 & 0.060 & 0.019 & 1604 & \textbf{ 35.61 }\\ \hline
c17 & 0.000 & 0.000 & 7 & 0.000 & 0.000 & 1 & \textbf{ 85.71}\\ \hline
c2670 & 0.024 & 0.007 & 876 & 0.019 & 0.006 & 514 & \textbf{ 41.32 }\\ \hline
c499 & 0.003 & 0.001 & 91 & 0.002 & 0.001 & 49 & \textbf{ 46.15 }\\ \hline
c880 & 0.009 & 0.003 & 331 & 0.007 & 0.002 & 171 & \textbf{ 48.34 }\\ \hline
c1355 & 0.014 & 0.004 & 507 & 0.010 & 0.003 & 213 &\textbf{  57.99} \\ \hline
c3540 & 0.033 & 0.010 & 1169 & 0.029 & 0.009 & 784 &\textbf{  32.93} \\ \hline
 & & & & & & & Average: 45.05 \\ \hline

 \multicolumn{8}{c}{\footnotesize \textbf{(c)}}\\
\end{tabular}
\caption{Experimental results on (a) arithmetic circuits, (b) cryptographic circuits, and (c) benchmark circuits of \cite{benchbf}. Computation costs and GC sizes are in terms of seconds and the number of generated ciphertexts, respectively. GC size reduction is calculated as the percentage of GC size reduction with respect to (privacy-free version of) half-gates garbling.}
\label{tbl:exp_arith}
\end{table}

\section{Implementation, Experimental Results, and Efficiency}\label{implem}

Our implementations of AuthOr as a proof-of-concept and of the state-of-the-art HG garbling \cite{ZRE15} for comparison, both in Python 3, are given at \cite{our_imple}. We have run tests on the arithmetic operation circuits of \cite{benchbf}, the cryptographic circuits of \cite{benchbf}, and the benchmark circuits of \cite{benchbch}, all of which were constructed by independent researchers. For each circuit, we have repeated our tests 10 times on an Apple M2 MacBook Pro with 24 GB of RAM. We present the average experimental results in Table \ref{tbl:exp_arith}, which compare costs for garbling and evaluating computation times (measured as elapsed times in seconds) and the GC size (measured as the total number of ciphertexts generated). As the bottleneck of garbling schemes is GC size \cite{ZRE15}, we specifically calculate the GC size reduction of our scheme in the last columns of the given tables. We did not measure the overhead associated with the compilation and preprocessing steps required before garbling.

\textbf{Preprocessing of the circuit.} In our implementation, we developed a compiler to parse circuits written in both supported syntaxes and to manage this emulation process. The circuits of \cite{benchbf} implement a buffer via an AND gate with both input pins connected to the same wire for some circuits. Since these gates do not affect the circuit’s mathematical functionality or the authenticity of garbling, we removed them in a preprocessing step. Additionally, some benchmarks include pins that are set as 0 through an XOR gate with a common input. For similar reasons, we also removed these XOR gates. Additionally, we employed the “networkx” library to determine the topological order of the gates in each Boolean circuit. Our implementation supports two syntaxes for representing Boolean circuits: \textit{Bristol Fashion} \cite{benchbf} and \textit{Bench}. In both formats, each circuit is represented by listing all internal Boolean gates and their pin connectivity. The set of gates discussed in Section \ref{gates} is Turing-complete, allowing us to emulate all other types of gates, including those with a fan-in greater than 2, as well as OR, NOR, NAND, and XNOR gates.

\textbf{Garbling.} The types of wires can be determined during the forward phase, which our implementation takes advantage of. We made this separate in Table \ref{our_garbling} only for the sake of clarity. As suggested by \cite{BHKR13,ZRE15}, the hash function $H$ can be implemented by using a symmetric-key encryption algorithm with a shared fixed key $k$. We use AES-128 encryption in CBC mode (inherited from the “Cryptography” library) to set the hash function $H$ as:
$$H(i,W)=AES128_k(W)\oplus i\oplus W$$

\textbf{GC size comparison with HG garbling \cite{ZRE15}.} As shown in Table \ref{tbl:exp_arith}, the GC size reduction is highly dependent on the architecture of the circuits. For cryptographic circuits of \cite{benchbf}, the gain is marginal because the complexity of gate connections requires our scheme to garble almost all the gates with $\mathsf{GbHG2}$ during the forward phase. In contrast, for arithmetic circuits of \cite{benchbf} and for benchmark circuits of \cite{benchbf}, the gain is up to roughly 98\% and 86\%, and averages roughly 14\% and 45\%, respectively.
\begin{table}[h]
\centering
\scriptsize
\begin{tabular}{|l|}
\hline
\textbf{procedure} $\mathsf{Count}(f)$: \\
\hline
label each wire as $\mathsf{TypeM}$ or $\mathsf{TypeS}$ \\
initialize $(\mathsf{CountHG2}, \mathsf{CountHG0\&1}, \mathsf{CountITAND)}:=(0,0,0)$\\
\textbf{for} each gate $i \in f$ \{in topological order\} \textbf{do} \\
\quad \textbf{if} $i \in f.\mathsf{ANDGates}$ \textbf{then} \\
\quad\quad $\{a,b\} \gets i.\mathsf{GateInputs}$ \\
\quad\quad \textbf{if} $a.\mathsf{Type}=\mathsf{TypeF}$ and $b.\mathsf{Type}=\mathsf{TypeF}$ \textbf{then} \\
\quad\quad\quad $\mathsf{CountHG2}:=\mathsf{CountHG2}+1$ \\
\quad\quad \textbf{else if} $\exists p\in\{a,b\}$ s.t.\ $p.\mathsf{Type}=\mathsf{TypeF}\vee p.\mathsf{Type}=\mathsf{TypeM}$ \textbf{then} \\
\quad\quad\quad $\mathsf{CountHG0\&1}:=\mathsf{CountHG0\&1}+1$ \\
\quad\quad \textbf{else} \\
\quad\quad\quad $\mathsf{CountITAND}:=\mathsf{CountITAND}+1$ \\
\quad \textbf{if} gate $i$ has a $\mathsf{TypeM}$ or $\mathsf{TypeF}$ input wire \textbf{then} \\
\quad\quad \textbf{for} each $p\in\{a,b,i\}$ s.t.\ $p.\mathsf{Type}\neq\mathsf{TypeF}$ \textbf{do} \\
\quad\quad\quad label $p$ as $\mathsf{TypeF}$ \\
\textbf{return} $(\mathsf{CountHG2}, \mathsf{CountHG0\&1}, \mathsf{CountITAND})$ \\
\hline
\end{tabular}
\caption{Counting the AND gates with different garbling algorithms.}
\label{count_}
\end{table}

\textbf{Efficiency gain based on Circuit Structure.} It may be a concern to decide whether or how much would a circuit benefit from AuthOr before implementation. Generally speaking, the gain is based on the depth of the circuit as the shollow circuits benefit more from $\mathsf{HG1}$ (or $\mathsf{HG0}$), and on the abundance of wires with $\mathsf{typeS}$ at lower depths as the ones with higher number of such wires would benefit from IT garbling more. We highlight that our scheme never results in greater GC size than half gates, so for simplicity we avoid providing a formulas to calculate the gain. Instead, we propose altering the counting algorithm $\mathsf{Count}$ in Table \ref{count_} to determine the numbers $\mathsf{CountHG2}$, $\mathsf{CountHG0\&1}$, and $\mathsf{CountITAND}$ of AND gates that $\mathsf{HG2}$, $\mathsf{HG1}$ (or $\mathsf{HG0}$), and $\mathsf{ITAND}$ are applied, respectively. Then estimating the efficiency gain is a straightforward task.

%\vspace{-10pt}
\printbibliography

@inproceedings{MLR23, author = {Marcadet, G. and Lafourcade, P. and Robert, L.}, title = {RMC-PVC: A Multi-Client Reusable Verifiable Computation Protocol}, year = {2023}, booktitle = {ACM SAC '23} }

@InProceedings{CKKZ12,
author="Choi, S. G.
and Katz, J.
and Kumaresan, R.
and Zhou, H.-S.",
title="On the Security of the ``Free-XOR'' Technique",
booktitle="TCC '12",
year="2012",
}

@misc{our_imple,
	author = {A. Ajorian},
	note = {Accessible via \url{https://github.com/Ajorian/AuthOr/tree/main#}},
	title = {Implementation of AuthOr and Half Gates garbling schemes},
	}

@InProceedings{Kol05,
author="Kolesnikov, V.",
title="Gate Evaluation Secret Sharing and Secure One-Round Two-Party Computation",
booktitle="ASIACRYPT '05",
year="2005",
}

@InProceedings{KKKS15,
author="Kempka, C.
and Kikuchi, R.
and Kiyoshima, S.
and Suzuki, K.",
title="Garbling Scheme for Formulas with Constant Size of Garbled Gates",
booktitle="ASIACRYPT '15",
year="2015",
}

@inproceedings{CHJV15, author = {Canetti, R. and Holmgren, J. and Jain, A. and Vaikuntanathan, V.}, title = {Succinct Garbling and Indistinguishability Obfuscation for RAM Programs}, year = {2015}, booktitle = {ACM STOC '15} }

@InProceedings{AL18,
author="Ananth, P.
and Lombardi, A.",
title="Succinct Garbling Schemes from Functional Encryption Through a Local Simulation Paradigm",
booktitle="TCC '18",
year="2018",
}

@InProceedings{BGG+14,
author="Boneh, D.
and Gentry, C.
and Gorbunov, S.
and Halevi, S.
and Nikolaenko, V.
and Segev, G.
and Vaikuntanathan, V.
and Vinayagamurthy, D.",
title="Fully Key-Homomorphic Encryption, Arithmetic Circuit ABE and Compact Garbled Circuits",
booktitle="EUROCRYPT '14",
year="2014",
}

@inproceedings{GKP+13, author = {Goldwasser, S. and Kalai, Y. and Popa, R. A. and Vaikuntanathan, V. and Zeldovich, N.}, title = {Reusable garbled circuits and succinct functional encryption},  booktitle =  {ACM STOC '2013} }

@InProceedings{Agr17,
author="Agrawal, S.",
title="Stronger Security for Reusable Garbled Circuits, General Definitions and Attacks",
booktitle="CRYPTO '17",
year="2017",
}

@misc{benchbf,
	author = {D. Archer and V. Arribas Abril and  S. Lu and P. Maene and N. Mertens and D. Sijacic and N. Smart},
	note = {Retrieved 2025-04-10 from \url{https://nigelsmart.github.io/MPC-Circuits/}},
	title = {{'Bristol Fashion' MPC Circuits}},
	}

@misc{benchbch,
	author = {M. Jenihhin},
	note = {Retrieved 2025-04-10 from \url{https://pld.ttu.ee/~maksim/benchmarks/iscas85/bench/}},
	title = {Bench marks},
	}

@inproceedings{BHKR13,
	author = {Bellare, M. and Hoang, T. and Keelveedhi, S. and Rogaway, P.},
	year = {2013},
	title = {Efficient Garbling from a Fixed-Key Blockcipher},
	booktitle = {IEEE SP '13}
}

@article{NS23,
author = {Nieminen, R. and Schneider, T.},
year = {2023},
month = {08},
pages = {},
title = {Breaking and Fixing Garbled Circuits When a Gate has Duplicate Input Wires},
volume = {36},
journal = {Journal of Cryptology},
}

@inproceedings{LWYY24,
	author = {H. Liu and X. Wang and K. Yang and Y. Yu},
	year = {2024},
	title = {Garbled Circuits with 1 Bit per Gate},
	booktitle = {Cryptology {ePrint} Archive, Paper 2024/1988},
}

@inproceedings{GLNP15, author = {Gueron, S. and Lindell, Y. and Nof, A. and Pinkas, B.}, title = {Fast Garbling of Circuits Under Standard Assumptions}, year = {2015}, booktitle = {ACM CCS '15} }

@misc{SL23,
      title={Quantum Analog of Shannon's Lower Bound Theorem}, 
      author={S. Basu and L. Parida},
      year={2023},
      url={https://arxiv.org/abs/2308.13091}, 
}

@InProceedings{CGM16,
author="Chase, M.
and Ganesh, C.
and Mohassel, P.",
title="Efficient Zero-Knowledge Proof of Algebraic and Non-Algebraic Statements with Applications to Privacy Preserving Credentials",
booktitle="CRYPTO '16",
year="2016",
}

@inproceedings{BHR12, author = {Bellare, Mihir and Hoang, Viet Tung and Rogaway, Phillip}, title = {Foundations of garbled circuits}, year = {2012}, booktitle = {ACM CCS '12} }

@InProceedings{LP11,
author="Lindell, Y.
and Pinkas, B.",
title="Secure Two-Party Computation via Cut-and-Choose Oblivious Transfer",
booktitle="TCC '11",
year="2011",
}

@InProceedings{RR21,
author="Rosulek, M.
and Roy, L.",
title="Three Halves Make a Whole? Beating the Half-Gates Lower Bound for Garbled Circuits",
booktitle="CRYPTO '21",
year="2021",
}

@InProceedings{DG17,
author="D{\"o}ttling, N.
and Garg, S.",
title="Identity-Based Encryption from the Diffie-Hellman Assumption",
booktitle="CRYPTO '17",
year="2017",
}

@inproceedings{JKO13, author = {Jawurek, M. and Kerschbaum, F. and Orlandi, C.}, title = {Zero-knowledge using garbled circuits: how to prove non-algebraic statements efficiently}, year = {2013}, booktitle = {ACM CCS '13},}

@InProceedings{FNO15,
author="Frederiksen, T. K.
and Nielsen, J. B. 
and Orlandi, C.",
title="Privacy-Free Garbled Circuits with Applications to Efficient Zero-Knowledge",
booktitle="EUROCRYPT '15",
year="2015",
}

@inproceedings{LR15, author = {Lindell, Y. and Riva, B.}, title = {Blazing Fast 2PC in the Offline/Online Setting with Security for Malicious Adversaries}, year = {2015},  booktitle = {ACM CCS '15} }

@inproceedings{ZRE15,
	author = {Zahur, S. and Rosulek, M. and Evans, D.},
	booktitle = {EUROCRYPT '15},
	title = {Two Halves Make a Whole: Reducing Data Transfer in Garbled Circuits using Half Gates},
	year = {2015}}

@inproceedings{Yao82,
	author = {Yao, A. C.},
	booktitle = {IEEE FOCS '82},
	title = {{Protocols for Secure Computations}},
	year = {1982}}

@inproceedings{Yao86,
	author = {Yao, A. C.},
	booktitle = {IEEE FOCS '86},
	title = {How to generate and exchange secrets},
	year = 1986}

@inproceedings{WRK17, author = {Wang, X. and Ranellucci, S. and Katz, J.}, title = {Global-Scale Secure Multiparty Computation}, year = {2017},  booktitle = {ACM CCS '17} }

@InProceedings{LPSY15,
author="Lindell, Y.
and Pinkas, B.
and Smart, N.
and Yanai, A.",
title="Efficient Constant Round Multi-party Computation Combining BMR and SPDZ",
booktitle="CRYPTO '15",
year="2015",
}

@InProceedings{GGP10,
author="Gennaro, R.
and Gentry, C.
and Parno, B.",
title="Non-interactive Verifiable Computing: Outsourcing Computation to Untrusted Workers",
booktitle="CRYPTO '10",
year="2010",
}

@InProceedings{GKPS18,
author="Ganesh, C.
and Kondi, Y.
and Patra, A.
and Sarkar, P.",
title="Efficient Adaptively Secure Zero-Knowledge from Garbled Circuits",
booktitle="PKC '18",
year="2018",
}

@InProceedings{KP17,
author="Kondi, Y.
and Patra, A.",
title="Privacy-Free Garbled Circuits for Formulas: Size Zero and Information-Theoretic",
booktitle="CRYPTO '17",
year="2017",
}

@inproceedings{KS08,
	author = {Kolesnikov, V. and Schneider, T.},
	booktitle = {ICALP '08},
	title = {Improved Garbled Circuit: Free XOR Gates and Applications},
	year = {2008}}

@InProceedings{GGPR13,
author="Gennaro, R.
and Gentry, C.
and Parno, B.
and Raykova, M.",
title="Quadratic Span Programs and Succinct NIZKs without PCPs",
booktitle="EUROCRYPT '13",
year="2013",
}

@ARTICLE{Sha49,
  author={Shannon, C. E.},
  journal={The Bell System Technical Journal}, 
  title={The synthesis of two-terminal switching circuits}, 
  year={1949},
  volume={28},
  number={1},
  pages={59-98},}

@inproceedings{Lindell2007,
	author = {Y. Lindell and B. Pinkas},
	booktitle = {EUROCRYPT '07},
	title = {An Efficient Protocol for Secure Two-Party Computation in the Presence of Malicious Adversaries},
	year = {2007}}
\appendix

\section{Inefficiency of IT Garbling for Generic Circuits} \label{inefficiency_of_IT}

For fairness of comparison with IT garbling \cite{KP17}, we convert an example generic circuit (with size $g$) given in Figure \ref{ExampleCircuit} into a Boolean formula by applying the technique mentioned in \cite{KP17} to convert a circuit into a Boolean formula: when an input wire $i$ of the circuit is input to $n$ gates, we replace $i$ by $n$ wires. Note that each circuit layer has 3 gates, so the depth of the circuit is $g/3$. Starting from the last layer in inverse topological order, the procedure replaces each fan-out-$n$ gate $i$ with $n$ fan-out-1 gates, each of which has the same input wires as the gate $i$ had. The last layer (the $(g/3)$-th) remains the same as it does not have any fan-out-2 gates. The $(g/3-1)$-th layer has 3 fan-out-2 gates, so is replaced with 6 gates. However, now the $(g/3-2)$-th layer has 3 fan-out-4 gates, so it needs to be replaced by 12 gates. At each layer the number of the gates doubles, hence we calculate the size $g'$ of the obtained Boolean formula as $g'= \sum_{j=0}^{g/3-1} 3\cdot 2^j=3\cdot 2^{g/3}-3$, which means $g'\in \Theta(\exp(g))$.

\begin{figure}[t]
\centering
\includegraphics[width=11cm]{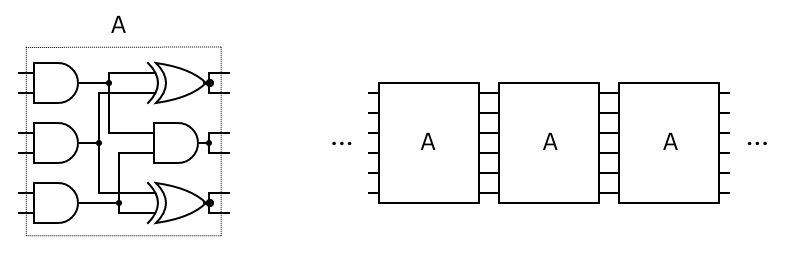}
\caption{An example circuit with size $g$ obtained by repeating $\mathsf{A}$-boxes $g/6$ times. The $\mathsf{A}$-box construction is also shown. Circuits are left to right in topological order.}
\label{ExampleCircuit}
\end{figure}

\section{Security of AuthOr} \label{sec_proof}

In what follows we show the authenticity and verifiability properties of AuthOr, as correctness simply follows from the correctness of the gate-garbling schemes and the use of the same wire labeling for $\mathsf{Gb}$ and $\mathsf{Ev}$.
\begin{theorem}
The AuthOr garbling scheme satisfies authenticity if the hash function $H$ has circular correlation robustness (CCR).
\end{theorem}
\begin{proof}
The proof that we provide here depends on the already proven security of the previous schemes: forward-garbled (FreeXOR and HG) and backward-garbled (IT). During the forward phase of $\mathsf{Gb}$, if the keys for a wire $i$ have been determined (turning that wire into $\mathsf{TypeF}$), then the keys for other input and output wires of a gate that takes $i$ as input are also determined (set as $\mathsf{TypeF}$). We deduce that at the end of a $\mathsf{Gb}$ execution, the $\mathsf{TypeF}$ wires look like “peninsulas” leaning toward the output of the circuit in a sea of $\mathsf{TypeB}$ wires, which is visualized in Figure \ref{peninsulas}.

The $\mathsf{TypeB}$ wire sea is already proven to leak no information (including XOR, AND, and NOT gates) about the other output-wire keys than the expected ones to be obtained during $\mathsf{Ev}$ by \cite{KP17}. Regarding the $\mathsf{TypeF}$ wire peninsulas, they are a simple extension HG$^*$ of the privacy-free HG garbling scheme of \cite{ZRE15} that we define here. HG$^*$ garbles only in the forward direction, including $\mathsf{FwNOT}$, which is already shown to have authenticity by \cite{Kol05}, $\mathsf{HG2}$ of \cite{ZRE15}, our proposals $\mathsf{HG0}$ and $\mathsf{HG1}$, and the XOR-garbling algorithms $\mathsf{FreeXOR0}$, $\mathsf{FreeXOR1}$, and $\mathsf{FreeXOR2}$. In what follows, we show that this scheme itself has authenticity:

\begin{lemma}
HG$^*$ garbling scheme (with $\mathsf{HGz}$, $\mathsf{FreeXORz}$, and $\mathsf{FwNOT}$ for $z\in\{0,1,2\}$) has authenticity if the privacy-free HG garbling scheme of \cite{ZRE15} (with $\mathsf{HG2}$, $\mathsf{FreeXORz}$, and $\mathsf{GbFwNOT}$, without $\mathsf{HG0}$, $\mathsf{HG1}$) has authenticity.
\end{lemma}

\begin{proof}
We construct a PPT adversary $\mathcal{B}$ that can break the authenticity of HG garbling (played with a challenger $\mathcal{C}$), by using as a subroutine a PPT adversary $\mathcal{A}$ that breaks the authenticity of the HG$^*$ garbling scheme.
\begin{enumerate}
\item $\mathcal{B}$ starts the execution of $\mathcal{A}$. $\mathcal{A}$ picks a function $f$ and an input string $\hat{x}$ and outputs them to $\mathcal{B}$, which in turn outputs them to $\mathcal{C}$.
\item $\mathcal{B}$ obtains the challenge garbled circuit $\hat{F}$ and input $\hat{X}$ from $\mathcal{C}$.
\item $\mathcal{B}$ executes the evaluation algorithm of HG garbling to obtain the key $W_i$ of each wire $i$.
\item $\mathcal{B}$ finds the set $\hat{G}$ of AND gates that would be garbled with $\mathsf{HG0}$ and $\mathsf{HG1}$ in a normal execution of the garbling algorithm of HG$^*$. Here, we assume the ones that would be garbled with $\mathsf{HG1}$ would have the left wire $a$ as $\mathsf{TypeF}$\footnote{For other circuits, it is trivial to find an equivalent that fits into this criterion as AND gates are symmetric}.
\item $\mathcal{B}$ executes the following updates on $\hat{F}$ and $\hat{X}$.
\begin{algorithmic}
\State \textbf{for} each gate $i \in f$ \{in topo.\ order\} \textbf{do}   
\State \quad     \textbf{if} $i \in G$ \textbf{then}
\State \quad \quad $\{a, b\} \gets i.\mathsf{GateInputs}$
\State \quad \quad Set $W_b^*\gets W_b \oplus F_i$
\State \quad \quad Delete $F_i$ from $\hat{F}$
\State \quad \quad Set each $W_b\in \hat{X}$ as the freshly generated key $W_b^*$
\end{algorithmic}
The updated $\hat{F}$ and $\hat{X}$ are information-theoretically indistinguishable from what could be obtained in an honest execution of HG$^*$ (explained below).
\item $\mathcal{B}$ gives $\hat{F}$ and $\hat{X}$ to $\mathcal{A}$.
\item $\mathcal{A}$ gives the output $\hat{Y}$ to $\mathcal{B}$, which in turn gives it to $\mathcal{C}$.
\end{enumerate}

\begin{figure}[t]
\includegraphics[width=12cm]{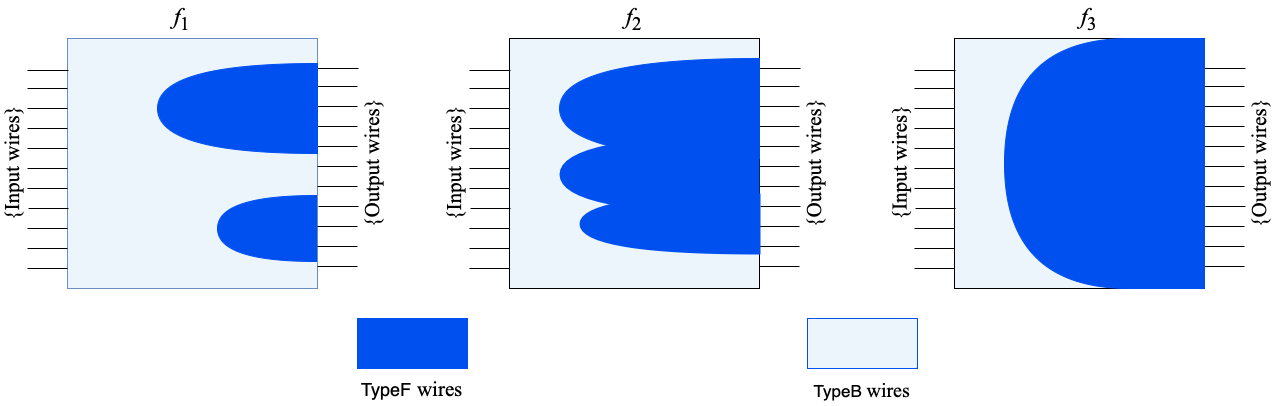}
\caption{Example formations of $\mathsf{TypeF}$ wires and $\mathsf{TypeB}$ wires upon $\mathsf{Gb}$ executions on example $f_1$, $f_2$, and $f_3$ (left to right in topological order).}
\label{peninsulas}
\end{figure}
In an honest execution of HG$^*$, for a given $W_a^0$, $\mathsf{GbHG1}$ sets $W_b^0$ as $H(i,W_a^0) \oplus H(i,W_a^1)$. Here, $\mathcal{B}$ additionally picks $F_i$ uniformly at random and sets $W_b^0$ as $H(i,W_a^0) \oplus H(i,W_a^1)\oplus F_i$. On the other hand, the HG garbling picks $W_b^0$ uniformly at random, and then uses $\mathsf{GbHG2}$ to set $F_i$ as $H(i,W_a^0) \oplus H(i,W_a^1)\oplus W_b^0$. The obtained $W_b^0$ and $F_i$ are information-theoretically indistinguishable from each other. The probability that $\mathcal{A}$ comes up with an authentic $\hat{Y}$ that is not equal to what could be obtained in an honest execution of the evaluation algorithm on the given $\hat{F}$ and $\hat{X}$ is exactly the same for $\mathcal{B}$ also doing so. Hence, if it is non-negligible, $\mathcal{B}$ breaks the authenticity of HG with non-negligible probability.
\end{proof}
So far, we have shown that the peninsulas of $\mathsf{TypeF}$ wires have authenticity given that the garbling scheme of \cite{ZRE15} has authenticity (which is based on the CCR of $H$ in their work), while \cite{KP17} already shows that the surrounding sea of $\mathsf{TypeB}$ wires has information-theoretic authenticity. To conclude, we now show that the IT garbling scheme of \cite{KP17} could be altered so that an arbitrary number of circuit output wires can have the same offset without breaking the information-theoretic authenticity. We name the altered scheme IT$^*$ garbling, which exactly has the garbling algorithms used in the backward phase of our $\mathsf{Gb}$. As in the generic IT garbling, IT$^*$ uses the $\mathsf{ITAND}$, $\mathsf{ITXOR}$, and $\mathsf{BwNOT}$ algorithms, and is applicable to circuits with only $\mathsf{TypeS}$ wires. Additionally, it allows the adversary to choose a set of output wires that will have the same offset, i.e., the XOR of the keys for each of those wires is fixed.
\begin{lemma}
The IT$^*$ garbling scheme is information-theoretically authentic.
\end{lemma}

\textsc{Proof.} Any circuit $f$ with only $\mathsf{TypeS}$ wires will be a set of independent subcircuits $\{f_1, \ldots, f_n\}$, each having a 1-bit output unconnected from each other. Each subcircuit $f_i$ that the adversary does not choose to have the same offset is garbled independently, which is proven to be information-theoretically authentic by \cite{KP17}. Hence, we will consider only the set of circuits $\{f'_1, \ldots, f'_m\}$ that the adversary chooses to have the same $\Delta$. Also, via the reduction from \cite{KP17}, the authenticity of each of those $f'_i$ can be reduced to the authenticity of a one-gate circuit, i.e., its output gate’s authenticity. Hence, we assume each of the circuits $f'_1, \ldots, f'_m$ has only one gate (it may be any of AND, XOR, or NOT) and has the same $\Delta$ at the output. We highlight that the IT$^*$ garbling scheme still has size zero, i.e., the adversary receives only $\hat{X}$ (specifically not $\hat{F}$) upon garbling. Therefore, the information-theoretic authenticity can be proven by proving that the same $\hat{X}$ may be generated with the same probability independent of $\Delta$ for any given circuit $f=\{f'_1, \ldots, f'_m\}$ where each $f'_i$ is a gate and its input $\hat{x}$ whose bits are $\mathsf{TypeS}$. We show this gate-by-gate for each of AND, XOR, and NOT gates.

\textbf{AND gate.} Given a gate $i$ with input wires $a$ and $b$, the garbling algorithm picks $W_i^0\twoheadleftarrow\{0,1\}^{\ell}$ uniformly at random, sets $W_i^1\gets W_i^0\oplus\Delta$ and then calls $\mathsf{GbITAND}$ to obtain the wire keys $W_a^0$, $W_a^1$, $W_b^0$, and $W_b^1$. Then, either $W_a^0$ or $W_a^1$ is sent as the wire key $X_a$ and either $W_b^0$ or $W_b^1$ is sent as the wire key $X_b$, depending on the values of $w_a$ and $w_b$. We will compare, for all combinations of $w_a$ and $w_b$, the probability of obtaining the same keys $X_a$ and $X_b$ using $\Delta_0$ with the one using $\Delta_1$.

\emph{Case $w_a=0$ and $w_b=0$.} The garbling algorithm picks $W_i^0$ uniformly at random, and sets $(W_a^0, W_b^0) \gets (W_i^0, W_i^0)$ via $\mathsf{GbITAND}$. Hence, any given pair of $X_a=W_a^0$ and $X_b=W_b^0$ (the given two keys for $w_a=0$ and $w_b=0$) could be generated with equal probability using $\Delta_0$ or $\Delta_1$.

\emph{Case $w_a=1$ and $w_b=0$.} The garbling algorithm picks $W_i^0\twoheadleftarrow \{0,1\}^\ell$ uniformly at random and sets $W_b^0 \gets W_i^0$ via $\mathsf{GbITAND}$. $\mathsf{GbITAND}$ also picks $W_a^1\twoheadleftarrow \{0,1\}^\ell$. Again, due to directly picking uniformly at random, any given pair of $X_a=W_a^1$ and $X_b=W_b^0$ could be generated with equal probability using either $\Delta_0$ or $\Delta_1$.

\emph{Case $w_a=0$ and $w_b=1$.} The garbling algorithm picks $W_i^0\twoheadleftarrow \{0,1\}^\ell$ uniformly at random, and sets $W_a^0 \gets W_i^0$ via $\mathsf{GbITAND}$. $\mathsf{GbITAND}$ also picks $W_a^1\twoheadleftarrow \{0,1\}^\ell$ uniformly at random, and sets $W_b^1\gets W_a^1 \oplus W_i^1$ where $W_i^1= W_i^0\oplus\Delta$. Any given pair of $X_a$ and $X_b$ could be generated with equal probability using either $\Delta_0$ or $\Delta_1$ due to the following reasons. Due to directly picking $W_i^0$ uniformly at random, $X_a=W_a^0=W_i^0$ could be generated with equal probability using either $\Delta_0$ or $\Delta_1$. Also, $W_b^1$ is generated by XORing the uniformly picked $W_a^1$ with some value $W_i^1$. Any element of $\{0,1\}^\ell$ could be generated as $W_b^1$ with equal probability regardless of $W_i^1$’s value, and then set as $X_b=W_b^1$.

\emph{Case $w_a=1$ and $w_b=1$.} The garbling algorithm picks $W_i^0\twoheadleftarrow \{0,1\}^\ell$ uniformly at random, and sets $W_i^1 \gets W_i^0\oplus\Delta$, hence $W_i^1$ could get any value from $\{0,1\}^\ell$ with equal probability using either $\Delta_0$ or $\Delta_1$. $\mathsf{GbITAND}$ then picks $W_a^1\twoheadleftarrow \{0,1\}^\ell$ uniformly at random, and sets $W_b^1\gets W_a^1 \oplus W_i^1$. We deduce $W_b^1=W_a^1 \oplus W_i^0\oplus \Delta$. Any given pair of $X_a$ and $X_b$ could be generated with equal probability using either $\Delta_0$ or $\Delta_1$ since $X_a$ directly comes from the uniform picking of $W_a^1$ and $X_b$ results from XORing some values with the uniformly picked $W_i^0$.

$W_a^0$ and $W_b^1$ are prepared independently, so their probabilities can be considered separately. Due to directly choosing $W_a^0$ uniformly, $X_a$ could be generated with equal probability using either $\Delta_0$ or $\Delta_1$. Due to the generation of $W_b^1$ by XORing the uniformly randomly picked $W_a^1$ with some value $W_i^1$, regardless of $W_i^1= W_i^0\oplus \Delta_0$ or $W_i^1= W_i^0\oplus \Delta_1$, any element of $\{0,1\}^\ell$ could be generated as $W_b^1$ with equal probability.

\textbf{XOR gate.} Similar to AND-gate garbling, given a gate $i$ with input wires $a$ and $b$, the garbling algorithm picks $W_i^0\twoheadleftarrow\{0,1\}^{\ell}$ uniformly at random, sets $W_i^1\gets W_i^0\oplus\Delta$ and then calls $\mathsf{GbITXOR}$ to obtain $W_a^0$, $W_a^1$, $W_b^0$, and $W_b^1$. Then, either $W_a^0$ or $W_a^1$ is sent as $X_a$ and either $W_b^0$ or $W_b^1$ is sent as $X_b$, depending on the values of $w_a$ and $w_b$. We will compare, for all combinations of $w_a$ and $w_b$, the probability of obtaining the same keys $X_a$ and $X_b$ using $\Delta_0$ with the one using $\Delta_1$.

\emph{Case $w_a=1$ and $w_b=1$.} The garbling algorithm picks $W_i^0\twoheadleftarrow \{0,1\}^\ell$ uniformly at random. $\mathsf{GbITXOR}$ then picks $W_a^1\twoheadleftarrow \{0,1\}^\ell$, and sets $W_b^1 \gets W_a^1 \oplus W_i^0$. As $W_a^1$ is picked uniformly and $W_b^1$ is randomized by the independent choice of $W_i^0$, any given pair of $X_a=W_a^1$ and $X_b=W_b^1$ could be generated with equal probability using either $\Delta_0$ or $\Delta_1$.

\emph{Case $w_a=1$ and $w_b=0$.} The garbling algorithm picks $W_i^0\twoheadleftarrow \{0,1\}^\ell$ uniformly at random, and sets $W_i^1\gets W_i^0 \oplus\Delta$. $\mathsf{GbITXOR}$ then picks $W_a^1\twoheadleftarrow \{0,1\}^\ell$, and sets $W_b^0 \gets W_a^1 \oplus W_i^1$, which means $W_b^0 = W_a^1 \oplus W_i^0 \oplus\Delta$. As $W_a^1$ is picked uniformly at random, and $W_b^0$ is randomized by the independent choice of $W_i^0$, any given pair of $X_a=W_a^1$ and $X_b=W_b^0$ could be generated with equal probability using $\Delta_0$ or $\Delta_1$.

\emph{Case $w_a=0$ and $w_b=1$.} The garbling algorithm picks $W_i^0\twoheadleftarrow \{0,1\}^\ell$ uniformly at random, and sets $W_i^1\gets W_i^0 \oplus\Delta$. $\mathsf{GbITXOR}$ then picks $W_a^1\twoheadleftarrow \{0,1\}^\ell$, and sets $W_b^1 \gets W_a^1 \oplus W_i^0$. It also sets $W_a^0\gets W_b^1 \oplus W_i^1$, which means $W_a^0= W_b^1 \oplus W_i^0 \oplus\Delta$. As $W_b^1$ is randomized by $W_a^1$, and $W_a^0$ is additionally randomized by the choice of $W_i^0$, any given pair of $X_a=W_a^0$ and $X_b=W_b^1$ could be generated with equal probability using either $\Delta_0$ or $\Delta_1$.

\emph{Case $w_a=0$ and $w_b=0$.} The garbling algorithm picks $W_i^0\twoheadleftarrow \{0,1\}^\ell$. $\mathsf{GbITXOR}$ then picks $W_a^1\twoheadleftarrow \{0,1\}^\ell$. Via XOR arithmetic, it is easy to show that, upon execution of $\mathsf{GbITXOR}$, $W_a^0= W_a^1 \oplus \Delta$ and $W_b^0 = W_a^1 \oplus W_i^0\oplus\Delta$. As $W_b^0$ is randomized by $W_a^1$, and $W_a^0$ is additionally randomized by the choice of $W_i^0$, any given pair of $X_a=W_a^0$ and $X_b=W_b^0$ could be generated with equal probability using either $\Delta_0$ or $\Delta_1$.

\textbf{NOT gate.} Given a gate $i$ with an input wire $a$, the garbling algorithm picks $W_i^0\twoheadleftarrow\{0,1\}^{\ell}$, sets $W_i^1\gets W_i^0\oplus\Delta$, and then swaps the order of the generated keys to obtain $W_a^0$ and $W_a^1$. As $W_a^1$ is set via directly picking $W_i^0$ uniformly at random, in the case of $w_a=1$, using either $\Delta_0$ or $\Delta_1$, $X_a$ could be obtained with equal probability. Also, as $W_a^0$ is randomized by $W_i^0$, in the case of $w_a=0$, again using either $\Delta_0$ or $\Delta_1$, $X_a$ could be obtained with equal probability.

We conclude that each of the gates $f'_1, \ldots, f'_n$ is garbled independently with independently picked randomizations, and the shown equality of probabilities that $X_a$ and $X_b$ are obtained from $\Delta_0$ and $\Delta_1$ extends to the whole set $\hat{X}$.
\end{proof}
\begin{theorem}
The AuthOr garbling scheme satisfies verifiability if the hash function $H$ has circular correlation robustness (CCR).
\end{theorem}
\begin{proof}[Intuition]
Our verifiability algorithm $\mathsf{Ve}$ calls gate verification algorithms from Table \ref{gate_verif} in topological order based on the gate-garbling schemes used. Each gate-garbling algorithm (except for $\mathsf{VeNOT}$) ensures that the output-wire key has only two possible assignments (unless the hash function $H$ is shown to be insecure), one for truth value 0 and one for truth value 1. $\mathsf{VeXOR}$ does so by ensuring that the input-wire offsets are the same, which means $W_i^0=W_a^0\oplus W_b^0=W_a^1\oplus W_b^1$ and $W_i^1=W_a^0\oplus W_b^1=W_a^1\oplus W_b^0$. The verification algorithms for HG garbling also guarantee that the input-wire offsets $\Delta$ are the same, in addition to the fact that $W_i^0=H(W_a^0)=W_b^0\oplus H(i,W_a^1)$ holds ($F_i$ is added for $\mathsf{HG2}$). As the inputs have the same offset $\Delta$, the XOR arithmetic ensures $W_i^1=W_i^0\oplus\Delta$. We have inherited $\mathsf{VeITAND}$ from \cite{KP17}, which has shown that the IT garbling scheme is verifiable. Also, all of the given gate verification algorithms return the two output-wire keys of each gate until the end of the circuit. Hence, for output gates, there exist only two possible assignments. Regarding the first probability of verifiability, this ensures that the obtained wire keys will be equal, and this probability is always 0. Regarding the second probability of verifiability, we can modify $\mathsf{Ve}$ to build an $\mathsf{Ext}$. Upon obtaining all the output-wire keys using $\mathsf{Ve}$, $\mathsf{Ext}$ finds the corresponding output keys to the bits of $y$ and returns them. Again, this ensures that the obtained wire keys will be equal to the ones that come from evaluation, and this probability is always 0.
\end{proof}

\end{document}